\documentclass[a4paper]{article}
\usepackage{amsmath}
\usepackage{amssymb}
\usepackage{amsthm}
\usepackage{caption2}
\usepackage{multirow}
\usepackage{cite}
\usepackage{bbding}
\usepackage{cancel}
\usepackage{color}
\usepackage{algorithm} %format of the algorithm
\usepackage{algorithmic} %format of the algorithm
\usepackage{geometry}
\usepackage{appendix}
\newtheorem{thm}{Theorem}[section]
\newtheorem{remark}[thm]{Remark}
\newtheorem{definition}[thm]{Definition}
\newtheorem{lem}[thm]{Lemma}
\newtheorem{prop}[thm]{Proposition}

\newtheorem{claim}[thm]{Claim}

\newcommand{\mb}{\mathbb}
\newcommand{\mc}{\mathcal}
\newcommand{\lc}{\lceil}
\newcommand{\rc}{\rceil}

\def\q#1{{\mathbb{Q}}_{#1}}

\begin{document}
\title{Optimal Quaternary Locally Repairable Codes Attaining the Singleton-like Bound}

\author{Yuanxiao Xi$^1$,~ Xiangliang Kong$^2$ and Gennian Ge$^2$}

\date{}
\maketitle
\footnotetext[1]{Y. Xi ({\tt yuanxiao\_xi@zju.edu.cn}) is with the School of Mathematical Sciences, Zhejiang University, Hangzhou 310027, Zhejiang, China.}
\footnotetext[2]{X. Kong ({\tt 2160501011@cnu.edu.cn}) and G. Ge ({\tt gnge@zju.edu.cn}) are with the School of Mathematical Sciences, Capital Normal University, Beijing 100048, China. The research of G. Ge was supported by the National Natural Science Foundation of China under Grant Grant No. 11971325 and Beijing Scholars Program.}

\begin{abstract}
  \indent Recent years, several new types of codes were introduced to provide fault-tolerance and guarantee system reliability in distributed storage systems, among which locally repairable codes (LRCs for short) have played an important role.

  A linear code is said to have locality $r$ if each of its code symbols can be repaired by accessing at most $r$ other code symbols. For an LRC with length $n$, dimension $k$ and locality $r$, its minimum distance $d$ was proved to satisfy the Singleton-like bound $d\leq n-k-\lc k/r\rc+2$. Since then, many works have been done for constructing LRCs meeting the Singleton-like bound over small fields.

  In this paper, we study quaternary LRCs meeting Singleton-like bound through a parity-check matrix approach. Using tools from finite geometry, we provide some new necessary conditions for LRCs being optimal. From this, we prove that there are $27$ different classes of parameters for optimal quaternary LRCs. Moreover, for each class, explicit constructions of corresponding optimal quaternary LRCs are presented.\\


%\noindent {{\it AMS subject classifications\/}:  05D05.}

\end{abstract}

\section{Introduction}\label{section_1}
%Modern large scale distributed storage systems, such as the data centers, store data in a redundant form to ensure reliability against node failures. The simplest and most commonly use technique is 3-replication, which has advantages due to its simplicity and fast recovery from data failures. However, this strategy entails large storage overhead and therefore is nonadaptive for storage systems supporting the explosion in the amount of data stored online.

Modern distributed storage systems have been transitioning to erasure coding based schemes with good storage efficiency in order to cope with the explosion in the amount of data stored online. Locally Repairable Codes (LRCs) have emerged as the codes of choice for many such scenarios and have been implemented in a number of large scale systems, for examples, Windows Azure \cite{HSX2012}, Facebook's Hadoop cluster \cite{SAP2013}.

The concept of codes with locality was introduced by Gopalan et al. \cite{GHS2012}, Oggier and Datta \cite{OFD2011}, and Papailiopoulos et al. \cite{PLD2012}. The $i_{th}$ coordinate of a code is said to have locality $r$ if it can be recovered by accessing at most $r$ other coordinates. LRCs are capable of very efficient erasure recovery for the typical case in distributed storage systems where a single node fails, while still allowing the recovery of data from a larger number of erasures.

A Singleton-type bound for locally repairable codes relating its length $n$, dimension $k$, minimum distance $d$ and locality $r$ was first shown in the highly influential work \cite{GHS2012}. It states that a linear locally repairable code $\mathcal{C}$ must obey
\begin{flalign}\label{singletonbound}
d(\mathcal{C})\leq n-k-\lc\frac{k}{r}\rc+2,
\end{flalign}
which reduces to the classical Singleton bound when $r=k$. Later, the bound was generalized to vector codes and nonlinear codes in \cite{FMY2014},~\cite{PD2014}. Although it certainly holds for all LRCs, it is not tight in many cases. The tightness of bound $(\ref{singletonbound})$ was studied in \cite{SWS2014}, \cite{AZ2015}.

%Besides the Singleton-type bound, a new bound taking field size into consideration is derived in [], which is call Cadambe-Mazumdar (C-M) bound,
%\begin{flalign}
%k\leq\min\limits_{t\in\mathbb{Z}^+}[tr+k_{opt}^{(q)}(n-t(r+1),d)],
%\end{flalign}
%where $k_{opt}^{(q)}(n,d)$ is the largest opssible dimension of a code, for given field size $q$, code length $n$, and minimum distance $d$, and $t\leq\min\{\lc n/(r+1),\lc k/r\rc\}$. This bound is improved by zzf.

We say an LRC is optimal if it satisfies bound $(\ref{singletonbound})$ with equality for given parameters $n$, $k$, $d$ and $r$. Many works have been done for the constructions of optimal LRCs. For the case $(r + 1)\mid n$, LRCs are constructed explicitly in \cite{TP2016} and \cite{SR2013} by using Reed-Solomon codes and Gabidulin codes respectively. However, both constructions are built over a finite field whose size is an exponential function of the code length $n$. In \cite{TB2014}, for the same case $(r + 1)\mid n$, the authors constructed an optimal code over a finite field of size sightly greater than $n$ by using ``good'' polynomials. This construction can be extended to the case $(r+1) \nmid n$ with the minimum distance $d \ge n-k-\lc k/r\rc+1$ which is at most one less than the bound (1). In \cite{TB2016} and \cite{BT2017}, the authors generalized this idea to cyclic codes and algebraic geometry codes.

For the convenience of computer hardware implementation, LRCs over small alphabets are of particular interest. In 2016, based on a construction of quasi-random codes, Ernvall et al. \cite{EW2016} constructed optimal LRCs over a small alphabet. By studying the properties of the corresponding parity-check matrices, Hao and Xia \cite{HXP2016} gave high rate optimal LRCs with $q\ge r+2$ and minimum distances 3 and 4. Then, with the same parity-check matrix approach, Hao et al. \cite{HX2016,HXC2017} determined all possible parameters of optimal binary and ternary $r$-LRCs.

%In this paper, combining tools from finite geometry, we employ the parity-check matrix approach to study the classification for parameters of optimal $(n,k,r)$-LRCs over the quaternary field (finite field of order 4), and we obtain the following main result.

In this paper, we employ the parity-check matrix approach to study the classification for parameters of optimal $(n,k,r)$-LRCs over the quaternary field (finite field of order 4) and obtain the following main result.

\begin{thm}\label{main}
Let $r\ge 1$, $k>r$ and $d\ge 2$. There are 27 classes of optimal quaternary $(n,k,r)$ LRCs with minimum distance $d$ meeting the Singleton-like bound, whose parameters and parity-check matrices are listed as follows respectively
\begin{itemize}
\setlength{\itemsep}{2pt}
\setlength{\parsep}{0pt}
\setlength{\parskip}{0pt}

  \item $(n,k,r)=(k+\lc k/r\rc,k,r)$ with $k>r\ge 1$, $d=2$, $H$ in $(\ref{pcm_d2})$, $(\ref{H_notlow})$, and $\underline{H}$ in $(\ref{H_low})$;
  \item $(n,k,r)=(3s+3,2s+1,2)$ with $s\ge 2$, $d=3$, $H$ in $(\ref{3s+3})$;
  \item $(n,k,r)=(4s+3,3s+1,3)$ with  $s\ge 2$, $d=3$, $H$ in $(\ref{4s+3_2})$;
  \item $(n,k,r)=(5s+3,4s+1,4)$ with $2\leq s\leq 3$, $d=3$, $H$ in $(\ref{5s+3})$;
  \item $(n,k,r)=(6s+3,5s+1,5)$ with $s=2$, $d=3$, $H$ in $(\ref{6s+3})$;
  \item $(n,k,r)=(4s+4,3s+2,3)$ with  $s\ge 2$, $d=3$, $H$ in $(\ref{4s+4})$;
  \item $(n,k,r)=(r+4,r+1,r)$ with  $2\leq r\le 15$, $d=3$, $H$ in $(\ref{pcm_hamming_code_1_3})$;
  \item $(n,k,r)=(r+5,r+2,r)$ with  $3\leq r\leq 15$, $d=3$, $H$ in $(\ref{pcm_hamming_code_1_4})$;
  \item $(n,k,r)=(r+6,r+3,r)$ with  $6\leq r\leq 15$, $d=3$, $H$ in $(\ref{pcm_hamming_5})$;
  \item $(n,k,r)=(4s+4,3s+1,3)$ with  $s\ge 2$, $d=4$, $H$ in $(\ref{4s+4_3s+1})$;
  \item $(n,k,r)=(5s+4,4s+1,4)$ with  $s\ge 2$, $d=4$, $H$ in $(\ref{5s+4_4s+1})$;
  \item $(n,k,r)=(6s+4,5s+1,5)$ with  $s\ge 2$, $d=4$, $H$ in $(\ref{6s+4_5s+1_1})$, $(\ref{6s+4_5s+1_2})$, $(\ref{6s+4_5s+1_3})$;
 % \item $(n,k,r)=(2r+6,2r+1,r)$ with $6\leq r\leq 7$, $d=4$;
  \item $(n,k,r)=(7s+4,6s+1,6)$ with $2\leq s\leq 3$, $d=4$, $H$ in $(\ref{7s+4})$;
  \item $(n,k,r)=(8s+4,7s+1,7)$ with $s=2$, $d=4$, $H$ in $(\ref{8s+4})$;
  \item $(n,k,r)=(5s+5,4s+2,4)$ with  $s\ge 2$, $d=4$, $H$ in $(\ref{5s+5_4s+2})$;
  \item $(n,k,r)=(6s+5,5s+2,5)$ with  $s\ge 2$, $d=4$, $H$ in $(\ref{6s+5_5s+2})$;
  \item $(n,k,r)=(7s+5,6s+2,6)$ with $s=2$, $d=4$, $H$ in $(\ref{7s+5})$;
  \item $(n,k,r)=(6s+6,5s+3,5)$ with  $s\ge 2$, $d=4$, $H$ in $(\ref{6s+6_5s+3})$;
  \item $(n,k,r)=(2s+2,s,1)$ with  $s\ge 2$, $d=4$, $H$ in $(\ref{2s+2_s})$;
  \item $(n,k,r)=(r+5,r+1,r)$ with  $2\leq r\leq 11$, $d=4$, $H$ in $(\ref{r+5_r+1})$;
  \item $(n,k,r)=(r+6,r+2,r)$ with  $3\leq r\leq 11$, $d=4$, $H$ in $(\ref{r+6_r+2})$;
  \item $(n,k,r)=(r+7,r+3,r)$ with  $5\leq r\leq 7$, $d=4$, $H$ in $(\ref{r+7_r+3})$;
  \item $(n,k,r)=(2r+6,2r+1,r)$ with $2\leq r\leq 7$, $d=4$, $H_i~(2\leq i\leq 7)$ in $(\ref{2r+6_2r+1_1})$, $(\ref{2r+6_2r+1_2})$, $(\ref{2r+6_2r+1_3})$, $(\ref{2r+6_2r+1_4})$;
  \item $(n,k,r)=(2k+4,k,1)$ with  $k=2,3$, $d=6$, $H$ in $(\ref{2k+4})$;
  \item $(n,k,r)=(2k+6,k,1)$ with  $k=2,3$, $d=8$, $H$ in $(\ref{2k+6})$;
  \item $(n,k,r)=(n,k,k-1)$ with  $3\le k\le 6$, $5\le n-k\le 6$, $d=n-k$, $H$ in $(\ref{near_mds})$;
  \item $(n,k,r)=(3s+6,2s+1,2)$ with  $s=2,3$, $d=6$, $H$ in $(\ref{3s+6_1})$, $(\ref{3s+6_2})$;
\end{itemize}

\end{thm}

As a by-product, we also derive some necessary conditions to determine whether the optimal LRCs satisfying the parameter requirements exist by the finite geometry approach and classical coding theory.
%Moreover, for each class of these parameters, we present explicit constructions. Besides, we also derive some new necessary conditions to determine whether the optimal LRCs satisfying the parameter requirements exist.

The rest of the paper is organized as follows. In Section \ref{section_2}, we present some notations, definitions about LRCs and some known results on parity-check matrix approaches. We also include some notations and results from finite geometry. In Section \ref{section_3}, we give some new necessary conditions for the existence of optimal LRCs via the finite geometry approach. In Section \ref{section_4}, we prove Theorem \ref{main} by determining all the possible parameters that optimal quaternary LRCs can have. Some concluding remarks are included in Section \ref{section_5}.

\section{Preliminary}\label{section_2}
\subsection{Coding theory and locally repairable codes}
Firstly, we introduce some notations and terminologies which will be frequently used throughout the paper.
\begin{itemize}
  \item [1.] Let $q$ be a prime power, and $\mathbb{F_q}$ be a finite field with $q$ elements and $\mathbb{F}_q^\ast$ be the corresponding multiplicative group.
  \item [2.] Given a positive integer $n$, denote $[n]=\{1,2,\ldots,n\}$. For any vector $\mathbf{a}=(a_1,a_2,\ldots,a_n)\in \mathbb{F}_q^{n}$, let ${\rm supp}(\mathbf{a})=\{i\in [n] : a_i\neq 0\}$ and $wt(\mathbf{a})=|{\rm supp}(\mathbf{a})|$. Given $i\in[n]$, the $i_{th}$ column is said to \emph{be covered by} the vector $\mathbf{a}$ if $i\in {\rm supp}(\mathbf{a})$. Given a set $A\subseteq\mathbb{F}_q^n$, the $i_{th}$ column is said to \emph{be covered by} $A$, if there exists $\mathbf{a}\in A$, such that the $i_{th}$ is covered by $\mathbf{a}$.
  \item[3.] Consider two different vectors $\mathbf{a}$ and $\mathbf{b}$, the \emph{Hamming distance} $d_{H}(\mathbf{a},\mathbf{b})$ is the number of coordinates at which they differ. For a code $\mathcal{C}\subseteq \mathbb{F}_q^{n}$, the minimum distance $d$ is the minimum value of distance between any two different codewords in $\mathcal{C}$.
  \item[4.] $\mathcal{C}$ is said to be an $[n,k,d]_q$ code if $\mathcal{C}$ is a linear code over $\mathbb{F}_q$ that has length $n$, dimension $k$, and minimum distance $d$. Usually we omit $q$ if $q$ is known. Given an $[n,k,d]$ linear code $\mathcal{C}$, its generator matrix is a $k\times n$ matrix $G$ whose rows form a basis for $\mathcal{C}$ and its parity-check matrix $H$ is an $(n-k)\times n$ matrix satisfying $GH^T=0$. Denote $\mathcal{C}^\bot$ as the dual code of $\mathcal{C}$, then the rows of $H$ form a basis for $\mathcal{C}^\bot$.
  \item[5.] Given matrices $A$ and $B$, denote $A\otimes B$ as the Kronecker product of $A$ and $B$.
\end{itemize}
%Let $\mathbb{F}_q$ be a finite field with $q$ elements and $\mathbb{F}_q^\ast$ be the corresponding multiplicative group. Denote $[n]:=\{1,2,\ldots,n\}$.  Consider two different vectors $\mathbf{a}$ and $\mathbf{b}$, the \emph{Hamming distance} $d_{H}(\mathbf{a},\mathbf{b})$ is the number of coordinates at which they differ. For a code $\mathcal{C}\subseteq \mathbb{F}_q^{n}$, the minimum distance $d$ is the minimum value of distance between any two different codewords in $\mathcal{C}$. In this paper, we focus on the quaternary field, $\mathbb{F}_4=\{0,1,\omega,\omega^2\}$, where $\omega$ is the primitive element of $\mathbb{F}_4$.

%An $[n,k,d]$ code $\mathcal{C}$ is a linear code over $\mathbb{F}_q$ that has length $n$, dimension $k$, and minimum distance $d$.

To present the formal definition of locally repairable codes, first, we need the concept of \emph{locality} for the symbol of codewords.

\begin{definition}\cite{GHS2012}\label{defn}
For $i\in[n]$, a code symbol $c_i$ of an $[n,k,d]$ linear code $\mathcal{C}$ is said to have locality $r$ if there exists a subset $R_i\subset[n]\setminus\{i\}$, $|R_i|\leq r$ such that $c_i$ can be recovered from the code symbols indexed by $R_i$, i.e.
$$c_i=\sum\limits_{i\in R_i}\lambda_{i,j}c_j$$
for some non-zeros $\lambda_{i,j}$ from the underlying field.
Equivalently, there exists a codeword $\textbf{h}_i$ in the dual code $\mathcal{C}^\bot$ such that $i\in supp(\textbf{h}_i)$ and $wt(\textbf{h}_i)\leq r+1$.
\end{definition}

Now, we give the definition of locally repairable codes and optimal locally repairable codes we consider in this paper.

\begin{definition}\cite{GHS2012}
For positive integers $n$, $k$ and $r$, an $(n,k,r)$ locally repairable code (LRC) is an $[n,k]$ linear code with locality $r$ for all symbols, i.e., each code symbol can be repaired by accessing at most $r$ other code symbols. Besides, the LRCs are called optimal if they meet the Singleton-like bound (\ref{singletonbound}).
\end{definition}

%For an $[n,k]$ linear code $\mathcal{C}$, it is called an $(n,k,r)$ LRC with information locality if only the $k$ information symbols have locality $r$ (usually we assume $r\ll k$). And it is called an $(n,k,r)$ LRC with all symbol locality, if all the $n$ code symbols have locality $r$. In this paper, we focus on LRCs with all symbol locality.

%From the view of the generator matrix, symbol $c_i$ has locality $r$ if and only if the $i$th column of $G$ is a linear combination of at most $r$ other column vectors.

For a linear code $\mathcal{C}$, there is a simple equivalency between the minimum distance $d$ and the linear dependency of columns of the parity-check matrix $H$.

\begin{lem}\cite[Corollary 1.4.14]{H2010}
A linear code has minimum distance $d$ if and only if its parity-check matrix $H$ has a set of $d$ linearly dependent columns but no set of $d-1$ linearly dependent columns.
\end{lem}

In \cite{HXP2016}, Hao and Xia use a parity-check approach for constructing LRCs: By Definition \ref{defn}, one can select $n-k$ specific codewords from $\mathcal{C}^\bot$ to form the parity-check matrix $H$. $H$ is divided into two parts:
\begin{flalign}\label{WLOG}
H=\left(\begin{array}{c} H_1\\H_2\end{array}\right),
\end{flalign}
where $H_1$ is an $l\times n$ matrix, $H_2$ is an $(n-k-l)\times n$ matrix. Rows in $H_1$, or $locality$-$rows$, cover all the columns to ensure the locality. As the assistant part, rows in the lower part $H_2$ ensure the minimum distance. The procedures for constructing such $H$ are described in Algorithm 1 in \cite{HXP2016}. Moreover, they also show that
%election procedure, firstly, for the first coordinate, select a codeword from $\mathcal{C}^\bot$ with weight at most $r+1$ to cover it; then for the first uncovered coordinate, select another codeword from $\mathcal{C}^\bot$ with weight at most $r+1$ to cover it; repeating the procedure iteratively until all the $n$ coordinates are covered and $H_1$ is constructed. Then we select some other $n-k-l$ independent codewords from $\mathcal{C}^\bot$ to form the lower part $H_2$ and finish the construction of $H$. The details are given as follows
%\begin{algorithm} %??・¨?×??
%\caption{\cite{HXP2016} } % ??・¨μ?????
%\label{alg1} %??・¨μ?±???
%\begin{algorithmic}[1] %′?′|μ?[1]??????????・¨??μ??????°??????±?÷?
%\STATE Let $i=1$, $S_0=\{\}$.
%\WHILE{ $S_{i-1}\ne[n]$:}
%\STATE Pick $j\in[n]\setminus S_{i-1}$
%\STATE Choose $\mathbf{h}_i=\arg\min_{\mathbf{e}\in\mathcal{C}^\bot,e_j\ne 0}wt(\mathbf{e})$.
%\STATE Set $S_{i}=S_{i-1}\cup supp(\mathbf{h}_i)$
%\STATE $i=i+1$
%\ENDWHILE
%\STATE Set $l=i-1$. Set $H_1=\left(\begin{array}{c} \mathbf{h}_1\\ \vdots\\ \mathbf{h}_l\end{array}\right)$.
%\STATE Choose additional $n-k-l$ vectors from $\mathcal{C}^\bot$ such that $H_2=\left(\begin{array}{c} \mathbf{h}_{l+1}\\ \vdots\\ \mathbf{h}_{n-k}\end{array}\right)$ and $H=\left(\begin{array}{c} H_1\\H_2\end{array}\right)$ is an $(n-k)\times n$ full-rank matrix.
%\end{algorithmic}
%\end{algorithm}
\begin{flalign}\label{inequality}
\lc k/r\rc\leq\lc n/(r+1)\rc\leq l\leq n-k.
\end{flalign}

 In this paper, we also use the parity-check matrix approach for our construction. For optimal LRCs, we have the following property for their parity-check matrices.
%Besides, we use the term \emph{optimal LRCs} to denote LRCs meeting the Singleton-like bound (\ref{singletonbound}).

\begin{lem}\cite{HXP2016}\label{r_mid_k}
For an $(n,k,r)$ LRC with $d=n-k-\lc k/r\rc+2$, suppose $r\mid k$, then $(r+1)\mid n$ and the supports of the locality-rows in the parity-check matrix must be pairwise disjoint, and each has weight exactly $r+1$.
\end{lem}

The classical Singleton bound \cite{H2010} says that $d\leq n-k+1$ and $[n,k,d]$ codes $\mathcal{C}$ meeting this bound are called \emph{maximum distance separable (MDS) codes}. Besides, codes with parameters $[n,k,n-k]$ are called \emph{almost maximum distance separable (AMDS) codes}.

The following two lemmas describe some substructures of the parity-check matrix of an optimal LRC, and build connections between the optimal LRC and the (almost) MDS code.
\begin{lem}\cite{HX2016}\label{abount_H'}
Let $\mathcal{C}$ be an optimal $(n,k,r)$ LRC with minimum distance $d$ over $\mathbb{F}_q$. Let $H'$ be the $m'\times n'$ matrix obtained from $H$ by deleting any $\lc k/r\rc-1$ locality-rows and all the columns they covered. Then $H'$ has full rank and the $[n',k',d']$ linear code $\mc{C}'$ with parity-check matrix $H'$ is an MDS code with $d'=d$.
\end{lem}
\begin{lem}\label{abount_H''}
Let $\mathcal{C}$ be an optimal $(n,k,r)$ LRC with dimension $\lc k/r\rc \ge 2$ and minimum distance $d$ over $\mathbb{F}_q$ and . Let $H''$ be the $m''\times n''$ matrix obtained from $H$ by deleting any $\lc k/r\rc-2$ locality-rows and all the columns they covered. Then the $[n'',k'',d'']$ linear code $\mc{C}''$ with $H''$ as parity-check matrix is an almost MDS code with $d''=d$.
\end{lem}
\begin{proof}[Proof of Lemma \ref{abount_H''}]
Since $H$ has $l$ locality-rows, by $(\ref{inequality})$, we have
\begin{flalign}\label{lemma_inequality}
m''=n-k-(\lc k/r\rc-2)\ge 2
\end{flalign}
and $H''$ has at least two rows with weight at most $r+1$, which are locality-rows from $H$.
Let $\gamma'$ be the number of columns covered by the $\lc k/r\rc-2$ locality-rows. Since every locality-row has weight at most $r+1$, we have
\begin{flalign}\label{lemma_inequality1}
n''=n-\gamma'\ge n-(\lc k/r \rc-2)(r+1).
\end{flalign}
Then, combining $(\ref{lemma_inequality})$ and $(\ref{lemma_inequality1})$ with $k+2r>r\cdot\lc k/r\rc$, we have $n''>m''\ge 2$. By the classical Singleton bound,
\begin{flalign*}
d''\leq n''-k''+1=Rank(H'')+1\leq m''+1.
\end{flalign*}
Among these $n''$ columns of $H$, since the elements not in $H''$ are all zero, we have $d\leq d''$.  Similarly, we can also obtain $d''\leq d'$.
Since $\mc{C}$ is an optimal LRC with $d=n-k-\lc k/r\rc+2$, we have
\begin{flalign*}
d''\geq d=n-k-\lc k/r\rc+2=m''.
\end{flalign*}

Note that $d=d'=m'+1$ by Lemma \ref{abount_H'}, where $m'$ is the number of rows of $H'$ defined in Lemma \ref{abount_H'}. If $d''=m''+1$, then $m''= m'$, which is impossible.
Therefore, $d=m''=n''-k''=d''$, which implies that $Rank(H'')=m''$ and $\mathcal{C}''$ is an almost MDS code.
\end{proof}
\begin{remark}
Actually, Lemma \ref{abount_H''} can be generalized to the following form with a similar process.

Let $\mathcal{C}$ be an optimal $(n,k,r)$ LRC with dimension $\lc k/r\rc \ge s$ and minimum distance $d$ over $\mathbb{F}_q$. Let $\widehat{H}$ be the $\widehat{m}\times \widehat{n}$ matrix obtained from $H$ by deleting any $\lc k/r\rc-s$ locality-rows and all the columns they covered. Then $\widehat{H}$ is the parity-check matrix of an $[\widehat{n},\widehat{k},\widehat{d}]$ linear code with $\widehat{d}=\widehat{n}-\widehat{k}+2-s=d$.
\end{remark}
Furthermore, we need some results about MDS codes and almost MDS codes to help us determine all the possible parameters of the optimal quaternary LRCs.
\begin{thm}\label{MDS}\cite[Corollary 7.4.3]{H2010}
 Assume that there exists an $[n,k,d]$ MDS code $\mathcal{C}$ over $\mb{F}_q$,
\begin{itemize}
  {\item[1.]if $2\le k$, then $d=n-k+1\leq q$;}
  {\item[2.]if $k\le n-2$, then $k+1\leq q$.
   }
\end{itemize}
\end{thm}

\begin{lem}\cite[Theorem 8]{DM1996}\label{amds}
If $\mathcal{C}$ is an $[n,k,n-k]$ almost MDS code over $\mathbb{F}_q$, $q>3$ with $k\ge 3$, then $n-k< 2q-1$.
\end{lem}

Given a code $\mathcal{C}$ of length $n$, denote $A_i$ $(0\leq i\leq n)$ as the number of codewords of weight $i$ in $\mathcal{C}$, then the list $A_i$ $(0\leq i\leq n)$ is called the weight distribution for $\mathcal{C}$. The following theorem determines the weight distribution of MDS code.
\begin{thm}\label{MDS_weight_distribution}\cite[Theorem 7.4.1]{H2010}
Let $\mathcal{C}$ be an $[n,k,d]$ MDS code over $\mathbb{F}_q$.  The weight distribution of $\mathcal{C}$ is given by $A_0=1$, $A_i=0$ for $1\leq i<d$, and
\begin{flalign*}
A_i={n\choose i}\sum\limits_{j=0}^{i-d}(-1)^j{i\choose j}(q^{i+1-d-j}-1),
\end{flalign*}
for $d\leq i\leq n$, where $d=n-k+1$.
\end{thm}

Finally, we need the following lemma to help us determine the upper bound on the minimum distance for optimal LRCs.
\begin{lem}\cite{HSX2019}\label{singleton_defect}
Let $\mathcal{C}$ be an optimal $q$-ary $(n,k,r)$-LRC with minimum distance $d$ and dimension $k>r\ge 1$, then
\begin{align*}
d\leq\left\{\begin{array}{ll}q,&\text{if }r\nmid (k-1)\text{;}\\2q,&\text{if }r\mid (k-1).\end{array}\right.
\end{align*}
\end{lem}

\subsection{Some results on finite geometry}
Denote $PG(n,q)$ as the $n$-dimensional projective space over $\mathbb{F}_q$. We call the $point$ in $PG(n,q)$ as the all $1$-subspace in $\mathbb{F}_q^{n+1}$, and the $line$ in $PG(n,q)$ as the all $2$-subspace in $\mathbb{F}_q^{n+1}$. A $k$\emph{-cap} in $PG(n,q)$ is a set of $k$ points, no 3 of which are collinear. A $k$-cap in $PG(n,q)$ is called complete if it is not contained in a $(k+1)$-cap of $PG(n,q)$. For $n=2$, a $k$-cap is usually called a $k$-arc.

Moreover, we denote $m_2(n,q)$ as the size of the largest complete cap of $PG(n,q)$. We need the following results about $m_2(n,q)$ in this paper.
\begin{lem}[\cite{HT2016,DF2009}]\label{finite_geometry}
\begin{enumerate}
  \item[1.] $m_2(2,q)=\left\{\begin{array}{ll}q+1,&q\text{ odd};\\q+2,&q\text{ even}.\end{array}\right.$
  \item[2.] $m_2(3,q)=q^2+1~(q>2)$.
  %\item[3.] $m_2(4,4)=41$.
  %\item[4.] $m_2(5,4)\leq 153$.
\end{enumerate}
\end{lem}

In $PG(2,q)$, an arc is called an \emph{oval} if its size achieves $q+1$, when $q$ is even, an arc is called an \emph{hyperoval} if its size achieves $q+2$. A \emph{conic} in $PG(2,q)$ is the set of points $\langle(x_1,x_2,x_3)\rangle$ satisfying a homogeneous quadratic polynomial condition of the form $Q(x_1,x_2,x_3)=0$, where
\begin{flalign*}
Q(x_1,x_2,x_3)=ax_1^2+bx_2^2+cx_3^2+dx_1x_2+ex_1x_3+fx_2x_3.
\end{flalign*}
The \emph{nucleus} of a conic $\mathfrak{C}$ is the only common point of the intersection of all tangent lines, which intersect $\mathfrak{C}$ in one point.

The following lemma gives the structure of hyperoval in $PG(2,4)$.
% and the exact structure of ovoids when $q=4$.
%$m_2(3,q)$. In $PG(3,4)$, a set of points is called \emph{elliptic quadric} if they all satisfy the equation $x_1x_2+x_3^2+x_3x_4+wx_4^2=0$, where $x_i$ denotes the $i_{th}$ coordinate of point.
\begin{thm}\cite[Theorem~2.9]{TP2003}\label{hyperoval}
All the hyperoval in $PG(2,4)$ consist of one conic and its nuclues. All the oval in $PG(2,4)$ is a conic.
\end{thm}

In $PG(3,q)$, an arc is called an \emph{ovoid} if its size achieves $q^2+1$. The following lemma demonstrates the structure of ovoids.
\begin{lem}\cite[Theorem 26.2]{M2007}\label{property_ovoid}
Let $\mathcal{O}$ be an ovoid of $PG(3,q)$, and let $P\in\mathcal{O}$. Then there is a unique plane $\pi_P$ such that $\pi_P\cap\mathcal{O}=\{P\}$. Every plane $\pi$ meets $\mathcal{O}$ either in a single point or in an oval of $\pi$.
\end{lem}

The following lemma gives an upper bound on the size of the cap in $PG(3,q)$ containing some specific structures related to conics.

 %in the planes of $PG(3,q)$

\begin{lem}\cite[Theorem 2]{SB1959}\label{finite_geometry_1}
Take $q=2^h$ for some positive integer $h$, let $\pi_1$ and $\pi_2$ be two distinct planes of $PG(3,q)$ and $L$ be their line of intersection. Consider a nondegenerate conic $\mathfrak{C}_1$ in $\pi_1$ and a nondegenerate conic $\mathfrak{C}_2$ in $\pi_2$, both of the two conics touching $L$ at the same point $T$ and having the same nucleus $\mathcal{O}$ (necessarily situated on $L$ and distinct from $T$).

The points set $\mathfrak{C}_1\cup\mathfrak{C}_2\cup\mathcal{O}$ constitutes an incomplete $(2q+2)$-cap. Every $k$-cap containing this $(2q+2)$-cap can be obtained by aggregating to it some points conveniently chosen on the plane $\pi_3$, where $\pi_3$ contains the line $L$ (distinct from $\pi_1,\pi_2$). Then the number of $k$ satisfies the limitation $k\leq 3q+2$. When $k=3q+2$, this $k$-cap is a complete cap.
\end{lem}

\section{Some necessary conditions}\label{section_3}

In this section, from the perspective of parity-check matrix, we will prove some new necessary conditions for the existence of optimal LRCs which will be useful in this paper.

Consider an $(n,k,r)$ LRC $\mathcal{C}$ over $\mathbb{F}_q$ with $n=l(r+1)$ and $k=lr-u$, where $l\ge 1$, $r>u\ge 0$. Assume that the parity-check matrix $H$ of $\mathcal{C}$ has the following form:
\begin{flalign}\label{H}
H=\left(\begin{array}{c}\begin{array}{cccc}\overbrace{\begin{array}{c}1~1~\cdots~1\\ \\ \\ \\ \end{array}}^{r+1}&\overbrace{\begin{array}{c}\\ 1~1~\cdots~1\\ \\ \\ \end{array}}^{r+1}& \begin{array}{c}\\ \\  \ddots \\ \\
\end{array}&\overbrace{\begin{array}{c} \\ \\ \\ 1~1~\cdots~1\\  \end{array}}^{r+1}\end{array}\\ \hline H_2\end{array}\right)
\end{flalign}
where the first $l$ rows in $H$, denoted as $H_1$, correspond to the locality-rows with disjoint repair groups of size $r+1$ and the lower part $H_2$ is a $u\times l(r+1)$ matrix over $\mathbb{F}_q$.

%For convenience, we say the $i_{th}$ column $\mathbf{c}_{i}$ of $H$ is covered by a row $\mathbf{r}$ of $H$ if $i\in {\rm supp}(\mathbf{r})$.
Consider a locality-row of $H$ which covers $r+1$ columns $\{a_1,a_2,\ldots,a_{r+1}\}$. A vector $\mathbf{v}$ is said to \emph{be generated by} columns $\{a_{i_1},a_{i_2},\ldots,a_{i_s}\}$, where $i_{j}\in[r+1]$, if $\mathbf{v}=\sum_{j=1}^{s}b_ja_{i_j}$, for some $b_j\in\mathbb{F}_q^\ast$ satisfying $\sum_{j=1}^{s}b_j=0$.

\begin{lem}\label{keylemma}
Let $H$ be an $(l+u)\times l(r+1)$ matrix over $\mathbb{F}_q$ of the form $(\ref{H})$
\begin{itemize}
  \item[1.] If any 4 columns of $H$ are linearly independent, then \begin{flalign}\label{keylemma1}
      l\cdot{r+1\choose 2}\leq\frac{q^u-1}{q-1}.
      \end{flalign}
  \item[2.] If any 5 columns of $H$ are linearly independent, and $r\leq 4$ then
\begin{flalign}
l\cdot{r+1\choose 2}&+(q-2)\cdot {r+1\choose 3}\leq\frac{q^u-1}{q-1}.\label{ii_1}
%,~~~~~~~~\text{if }r\le 4
%l\cdot{r+1\choose 2}&+(q-1)\cdot {5\choose 3}+\sum\limits_{i=5}^{\min\{6,r\}}{i\choose 2}\leq\frac{q^u-1}{q-1},~~~~~\text{if }r\ge 5.\label{ii_2}
\end{flalign}
  \item[3.] If any 6 columns of $H$ are linearly independent, then for all $r$,
      \begin{align}
      l\cdot r&\leq m_2(u-1,q),\label{iii_1}\\
      l(q-2)\cdot{r+1\choose 3}&\leq\frac{q^u-1}{q-1}. \label{iii_2}
      \end{align}
%  \item[4.] If $q=4$ and any 7 columns of $H$ are linearly independent, then for all $r$,
%      \begin{flalign}\label{iv_1}
%      l&\cdot{r+1\choose 2}+2\cdot{\min{\{r+1,5\}}\choose 3}\leq m_2(u-1,q). %~~~~\text{if }r\le 4;\label{iv_1}\\
%      %l&\cdot{r+1\choose 2}+2\cdot {5\choose 3}\leq m_2(u-1,q),~~~~~\text{if }r\ge 5.\label{iv_2}
%      \end{flalign}
\end{itemize}
For completeness, we define ${n\choose k}:=0$, if $n<k$.
\end{lem}

\begin{proof}[Proof of Lemma \ref{keylemma}]
The idea of the proof is to convert sets of the columns of the parity-check matrix $H$ into points of a projective space and obtain a set of pairwise distinct points in this projective space from the matrix $H$. After that we shall use some results on the structure of projective spaces to give an upper bound for the number of such points to derive the above inequalities.

\textbf{Case 1.} This result was first proposed in \cite{HXC2017}, here we reprove it in a different way.

By the definition of $H$, the uppermost nonzero entry of each column lies in the corresponding locality-row. Given a locality-row with weight $r+1$, any two of columns covered by the same locality-row could generate $q-1$ different nonzero vectors, and each of these $q-1$ vectors is a multiple of any one of the rest. Since the first $l$ coordinates of these $q-1$ vectors are all $0$, thus we can omit these $l$ coordinates and regard these $q-1$ vectors as one point of $PG(u-1,q)$.

From this procedure, for all $l$ locality-rows of $H$, we can obtain $l\cdot{r+1\choose 2}$ points of $PG(u-1,q)$. Since any 4 columns of $H$ are linearly independent, all these points have to be pairwise independent, i.e. any point can not be a multiple of the other, which implies that these points are pairwise distinct. Since $PG(u-1,q)$ has $(q^u-1)/(q-1)$ distinct points, we obtain (\ref{keylemma1}).

\textbf{Case 2.} First, similar to the discussion of Case 1, we can obtain $l\cdot {r+1\choose 2}$ points of $PG(u-1,q)$, each point is generated by two different columns covered by a same locality-row.

Now, fix a locality-row with weight $r+1$, consider the $(q-2)(q-1)$ different nonzero vectors generated by $3$ distinct columns covered by the same locality-row. Due to the multiple relationship among these vectors, they can only be regarded as $(q-2)$ points of $PG(u-1,q)$. If $r\le 4$, since any 5 columns of $H$ are linearly independent, these $2\cdot {r+1\choose 3}$ points generated by $3$ distinct columns covered by the same locality-row have to be pairwise distinct. Moreover, the linear independency also guarantees that these $(q-2)\cdot{r+1\choose 3}$ points are disjoint from the $l\cdot {r+1\choose 2}$ points obtained above. Since $PG(u-1,q)$ has $(q^u-1)/(q-1)$ distinct points, we obtain (\ref{ii_1}).

\textbf{Case 3.} Similarly, for all $l$ locality-rows of $H$, we can obtain $lr$ points of $PG(u-1,q)$, each point is generated by one fixed column and other column covered by a same locality-row. Since any 6 columns of $H$ are linearly independent, we know that any 3 of these points are not collinear. Therefore, the number of these points is at most the largest size of the complete cap of $k$-caps in $PG(u-1,q)$, i.e. $m_2(u-1,q)$. Thus, we have (\ref{iii_1}).

%For each locality-row and $r+1$ columns its covered, if we consider the points generated by arbitrary 2 columns of them, we obtain at most $l\cdot{r+1\choose 2}$ points. Since that any 6 columns of $H$ are linear independent, any 3 of these points have to be pairwise independent, which implies that these points of projective of space are different and any three of them are not collinear, corresponding the $k$-cap in the projective space. Therefore we obtain (\ref{iii_1}).
Meanwhile, note that each locality-row covers $r+1$ columns, we can obtain $(q-2)\cdot{r+1\choose 3}$ points of $PG(u-1,q)$ generated by sets of 3 columns covered by a certain locality-row. Since any 6 columns of $H$ are linearly independent, any two of these points have to be pairwise independent. This implies that these points are pairwise distinct. Since $|PG(u-1,q)|\leq(q^u-1)/(q-1)$, we have (\ref{iii_2}).

\end{proof}

To simplify our discussion in Section \ref{section_4}, we need the following proposition.

\begin{prop}\label{prop_hamming}
Given positive integers $r\ge 2$, $p>3$ and $r\ge p-3$, let $\mathcal{C}$ be an optimal $(r+p,r+p-3,r)$ LRC with $d=3$ over $\mathbb{F}_4$, then $r+p\leq 21$. Moveover, if $H$ is a parity-check matrix of $\mathcal{C}$, then $H$ is a parity-check matrix of a quaternary $[21,18,3]$ Hamming code or its shortened version. Moreover, each of the locality-rows of $H$ has at least $p-1$ zeros and there exists one of the locality-rows of $H$ has exactly $p-1$ zeros.
\end{prop}
\begin{proof}
Since $d=3$, each column of $H$ can be regarded as a point in $PG(2,4)$. Hence, $r+p\leq (4^3-1)/(4-1)=21$. By \cite[Theorem 1.8.2]{H2010}, any quaternary $[21,18,3]$ code is equivalent to the Hamming code. When $5<r+p<21$, since the number of columns of $H$ is less than $21$, $H$ can be viewed as a parity-check matrix of a shorten version of the $[21,18,3]$ Hamming code. Moreover, each of the locality-rows of $H$ has at least $p-1$ zeros and there exists one of the locality-rows of $H$ has exactly $p-1$ zeros to satisfy the limitation of parameters.
\end{proof}

\section{The classification of the optimal quaternary LRCs}\label{section_4}

Based on the necessary conditions proved in Section \ref{section_3}, now we begin to determine all the parameters of optimal quaternary LRCs meeting the Singleton-like bound. For each class of parameters, we will present an explicit construction. We consider the quaternary field, $\mathbb{F}_4=\{0,1,\omega,\omega\}$, where $\omega$ is the primitive of $\mathbb{F}_4$

%\indent The analysis procedure is similar as \cite{HX2016,HXC2017}, which has enumerated all the binary optimal and ternary optimal LRCs.

The following proposition follows from Theorem \ref{MDS} and Lemma \ref{abount_H'}.

\begin{prop}\label{prop01}
Let $\mathcal{C}$ be an optimal quaternary $(n,k,r)$ LRC with $d=n-k-\lceil k/r\rceil+2$ according to $(\ref{singletonbound})$ and $H$ be its parity-check matrix of the form described in $(\ref{WLOG})$. Let $H'$ be the $m'\times n'$ matrix obtained from $H$ by deleting any fixed $\lceil k/r\rceil-1$ locality-rows and all the columns they covered. Then $H'$ is a full rank parity-check matrix of an $[n',k',d']$ linear code $\mathcal{C}'$ over $\mathbb{F}_4$ with the following possible parameters:

\centering
\begin{tabular}{|c|c|c|c|c|c|c|}
  \hline$n'$&$n'(n'\ge 2)$&$n'(n'\ge 3)$&$4$&$5$&$5$&$6$\\
  \hline$k'$&$n'-1$&$1$&$2$&$2$&$3$&$3$\\
  \hline$d'$&$2$&$n'$&$3$&$4$&$3$&$4$\\
  \hline
\end{tabular}
%$[n',n'-1,2]$ $(n'\ge2)$, $[n',1,n']$ $(n'\ge 3)$, $[4,2,3]$, $[5,2,4]$, $[5,3,3]$, $[6,3,4]$.
\end{prop}

Applying Lemma \ref{singleton_defect} with $q=4$, we can obtain that the minimum distance of a quaternary optimal LRC is at most $8$. Then according to the Proposition \ref{prop01}, our discussion will be divided into the following 4 cases.

1. $d=2$ and $H'$ contains exactly one row. $H'$ is a parity-check matrix of a quaternary $[n',n'-1,2]$ $(n'\ge 2)$ MDS code.

2. $d=3$ and $H'$ contains two rows. $H'$ is a parity-check matrix of a quaternary $[5,3,3]$ or $[4,2,3]$ or $[3,1,3]$ MDS code.

3. $d=4$ and $H'$ contains three rows. $H'$ is a parity-check matrix of a quaternary $[6,3,4]$ or $[5,2,4]$ or $[4,1,4]$ MDS code.

4. $5\leq d\leq 8$ and $H'$ contains more than three rows. $H'$ is a parity-check matrix of a quaternary $[n',1,n']$ $(n'\ge 5)$ MDS code.

\subsection{$d=2$ and $H'$ contains one row}\label{subsection_4_1}

%Clearly, $H'$ is a row vector with all the entries being 1 or $\omega$ or $\omega^2$. Moreover, since we obtain $H'$ by deleting any fixed $\lceil k/r\rceil$ locality rows and all the columns whose coordinates are covered by the supports of these locality-rows. And from xxx we know $\lceil k/r\rceil-1<l$ which is the number of locality-rows. so $H'$ has a locality-row with weight at most $r+1$. This implies that $H'$ has to be a row vector with length at most $r+1$, i.e. $n'\leq r+1$. Since $n'\ge 2$, $d'=d=2$ and $n=k+\lceil k/r\rceil$. Hence, $\mathcal{C}$ must be a quaternary $[k+\lceil k/r\rceil,k,2]$ with locality r. Moreover, by xx, $\lceil k/r\rceil=l=n-k$, which implies that $H$ consists of only locality-rows.\\
In this case, by Proposition \ref{prop01}, $H'$ is a full rank parity-check matrix of a quaternary $[n',n'-1,2]$ $(n'\ge2)$ MDS code. Thus we have $d=d'=2$, $m'=n-k-(\lc k/r\rc-1)=1$. Recall that $l$ is the number of the rows in $H_1$, by (\ref{inequality}), we have $\lc k/r\rc=l=n-k$, which means that $H=H_1$.

If $r\mid k$, then $n-k=k/r$ and $n=(r+1)k/r$, thus all $k/r$ rows of $H$ have weight $r+1$ and their supports are pairwise disjoint. The parity-check matrix of the $[k+k/r,k,2]$ LRC with locality $r$ has the following form:
\begin{flalign}\label{pcm_d2}
H=\big(I_{\frac{k}{r}}\otimes\underbrace{(1~1~\ldots~1)}_{r+1}\big)_{\frac{k}{r}\times\frac{(r+1)k}{r}}.
\end{flalign}
E.g., for $n=12$, $k=9$ and $r=3$, we have
\begin{flalign*}
H=\left(\begin{array}{cccccccccccc}1&1& 1 & 1 & 0 & 0 & 0 & 0 & 0 & 0 & 0 & 0 \\0&0& 0 & 0 & 1 & 1 & 1 & 1 & 0 & 0 & 0 & 0 \\0&0& 0 & 0 & 0 & 0 & 0 & 0 & 1 & 1 & 1 & 1 \end{array}\right).
\end{flalign*}

If $r\nmid k$, then $r\ge 2$. Let $k=sr+t$, where $1\leq t\leq r-1$, then we have $\lceil\frac{k}{r}\rceil=s+1$, $n=k+\lceil\frac{k}{r}\rceil=(r+1)\lceil\frac{k}{r}\rceil-(r-t)$, where $1\leq r-t\leq r-1$. Let $\hat{H}$ be a $\lceil\frac{k}{r}\rceil\times(r+1)\lceil\frac{k}{r}\rceil$ matrix of the form given in (\ref{pcm_d2}), where $\frac{k}{r}$ is replaced with $\lceil\frac{k}{r}\rceil$.

From the analysis above, we have the following two types of parity-check matrices.
\begin{flalign}
H \text{ is } &\text{a }\lceil\frac{k}{r}\rceil\times(k+\lceil\frac{k}{r}\rceil)\text{ matrix obtained from }\hat{H} \text{ by deleting any } r-t\nonumber\\ &\text{columns of }\hat{H}, \text{ such that at least one}\text{ row of }H\text{ has weight }r+1;\label{H_notlow}\\
\underline{H}\text{ is }&\text{obtained from }H\text{ by substituting at most }r-t~0\text{'s of }H\text{ to }1\text{'s or }\omega\text{'s}\nonumber\\&\text{or }\omega^2\text{'s such that the weight of each row of }\underline{H}\text{ is at most }r+1.\label{H_low}
\end{flalign}
\indent Then, in the sense of equivalence, every $(k+\lc\frac{k}{r}\rc,k,r)$ LRC with $d=2$ must have parity-check matrix as $H$ or $\underline{H}$. E.g. for $n=11$, $k=8$ and $r=3$, its parity-check matrix is
\begin{flalign*}
H=\left(\begin{array}{ccccccccccc}1&1& 1 & 1 & 0 & 0 & 0 & 0 & 0 & 0 & 0  \\0&0& 0 &0 & 1 & 1 & 1 & 1 & 0 & 0 & 0  \\0&0& 0 & \underline{0} & 0 & 0 & 0  & \underline{0} & 1 & 1 & 1 \end{array}\right).
\end{flalign*}
With any one of the two underlined zeros being substituted to $1$ or $\omega$ or $\omega^2$, $\underline{H}$ is obtained.
\begin{remark}
The analysis in this subsection is quite similar to that in \cite{HX2016,HXC2017}, since it has no relevance to the field size.
\end{remark}
\subsection{$d=3$ and $H'$ contains two rows}\label{d=3}

Since $H'$ is a parity-check matrix of a quaternary $[5,3,3]$ or $[4,2,3]$ or $[3,1,3]$ MDS code, we have $n-k-\lc k/r\rc=d-2=1$. Meanwhile, since $H_1$ has $l\ge\lc k/r\rc$ rows, by (\ref{inequality}), we have
\begin{flalign*}
\lc k/r\rc\leq l\leq n-k=\lc k/r\rc+1.
\end{flalign*}

\begin{itemize}
  \item For the case $l=\lc k/r\rc$
\end{itemize}

%\textbf{Case 1:}.
We have $n-k=l+1$, which implies that $H_2$ contains only one row. According to the construction procedure of $H'$, $H'$ contains a locality-row covering all the columns remained after the deletion. Therefore, consider the linear code $\mathcal{C'}$ of length $n'$ with $H'$ as its parity-check matrix, the dual code of $\mathcal{C}'$ has weight distribution with $A_{n'}>0$. By Theorem \ref{MDS_weight_distribution}, dual codes of both $[4,2,3]$ and $[3,1,3]$ MDS codes satisfy this condition, while the dual code of $[5,3,3]$ MDS code has weight distribution with $A_5=0$, thus $H'$ can be the parity-check matrix of a $[4,2,3]$ MDS code or a $[3,1,3]$ MDS code. This leads to $n'=3$ or $n'=4$.

If $r\mid k$, set $k=sr$ for some $s>0$. By Lemma \ref{r_mid_k}, we have $(r+1)\mid n$. Since $n=k+l+1=s(r+1)+1$, this contradicts the fact that $(r+1)\mid n$. Therefore, optimal LRCs with such parameters do not exist.

If $r\nmid k$, then $r\ge 2$. Let $k=sr+t$, where $1\le t\le r-1$ and we denote $\gamma$ as the number of columns covered by the supports of the deleted $\lc k/r\rc-1$ locality-rows. Thus, we have $\gamma\leq(\lc k/r\rc-1)(r+1)$. Then
\begin{flalign}\label{inequality_1}
k+\lc k/r \rc+1-(\lc k/r\rc-1)(r+1)\leq n-\gamma=n'\leq 4,
\end{flalign}
i.e.
\begin{equation*}
k-r\cdot\lc k/r\rc+r\le 2.
\end{equation*}
This leads to $t=1$ or $t=2$.

\textbf{Case} $\mathbf{t=1}$: In this case, we have $3\le n'\le 4$, which implies that $H'$ is the parity-check matrix of a quaternary $[4,2,3]$ or $[3,1,3]$ MDS code.

%From \cite{HXC2017}, if $H\neq H_1$ and $H'$ is the parity-check matrix of an $[n',1,n'-1]$ MDS code, then $\mathcal{C}$ is a near MDS code, when $s=1$. Therefore we first consider additional circumstances if $s=1$, $n'=3$. $\mathcal{C}$ is nearly MDS code. \\

If $s=1$, we have $k=r+1$, $n=r+4$ and $r\ge 2$. Notice that after removing arbitrary locality-row with $m$ zeros and the columns it covered, we can obtain a parity-check matrix of some code with length $m$. Therefore, each locality-row of $H$ must have $3$ or $4$ zeros to ensure the existence of $H'$. Moreover, applying Proposition \ref{prop_hamming} with $p=4$, we have $n=r+4\leq 21$, and one of locality-rows contains exactly 3 zeros. When $n=21$,
\begin{flalign}\label{pcm_hamming_code}
H=\left(\begin{array}{ccccccccccccccccccccc}
1&1&1&1&1&1&1&1&1&1&1&1&1&1&1&1&0&0&0&0&0\\
1&1&1&\omega&\omega&\omega&\omega^2&\omega^2&\omega^2&\omega^2&\omega&0&0&0&0&1&0&1&1&1&1\\
1&\omega&\omega^2&1&\omega&\omega^2&1&\omega&\omega^2&0&0&\omega^2&0&\omega&1&0&1&\omega^2&\omega&1&0
\end{array}\right).
\end{flalign}
One can easily show that $(\ref{pcm_hamming_code})$ doesn't meet the restrictions above. Similarly, when $n=20$, as a shortened version of the $[21,18,3]$ Hamming code, the $[20,17,3]$ code doesn't meet the restrictions either. Thus we have $n\leq 19$ and $r\leq 15$. When $n=19$, we have
\begin{flalign}\label{pcm_hamming_code_1_3}
H=\left(\begin{array}{ccccccccccccccccccccc}
1&1&1&1&1&1&1&1&1&1&1&1&1&1&1&1&0&0&0\\
1&1&1&\omega&\omega&\omega&\omega^2&\omega^2&\omega^2&\omega^2&\omega&1&0&0&0&0&1&1&1\\
\hline
1&\omega&\omega^2&1&\omega&\omega^2&1&\omega&\omega^2&0&0&0&\omega^2&0&\omega&1&\omega&1&0
\end{array}\right).
\end{flalign}
Therefore, we can puncture the first 0 to 13 columns from (\ref{pcm_hamming_code_1_3}) respectively to obtain an optimal LRC with parameter $[r+4,r+1,3]$, $2\leq r\leq 15$.
%
%Since $k\ge 3$, $n-k=3$, $q=4$, according to \citep[Theorem 3.5]{DSL1995}, we have $3\leq k\leq 2q=8$, therefore $2\le r\leq 7$.  When $n=11$, $k=8$, the parity check matrix is
%\begin{flalign}\label{near_mds_d3}
%H=\left(\begin{array}{ccccccccccc}
%1&1&1&0&0&0&1&1&1&1&\omega\\
%0&0&0&1&1&1&1&1&1&1&1\\
%\hline
%0&1&\omega^2&0&1&\omega^2&0&\omega&\omega^2&1&1
%\end{array}\right).
%\end{flalign}
%We obtain codes with parameter $[r+4,r+1,3]$ $(2\leq r\leq 6)$ by puncturing the last 1 to 5 columns respectively from (\ref{near_mds_d3}).\\

When $s\ge 2$, if $n'=3$, we have $\gamma=n-n'=s(r+1)$. Thus, the supports of the $s$ deleted locality-rows are pairwise disjoint with size $r+1$. Similarly, when $n'=4$, one can show that the supports of $s$ deleted locality-rows intersect in one column.
%
%one can show that among these $s$ deleted locality-rows, there exist two rows sharing at most one common coordinate in their supports.

W.l.o.g., suppose that the weight of the first locality-row of $H$ is $r+1$. When $s\ge 4$, we claim that $r\leq 3$. Otherwise, suppose $r\ge 4$ and there exists a column covered by the first two locality-rows. Due to the arbitrariness of $s$ deleted rows, all the last $s-2$ locality-rows are pairwise disjoint and have weight $r+1$. If we delete the first $s$ locality-rows and the coordinates they cover, the resulting $H'$ has length $r+1\ge 5$, a contradiction. On the other hand, if there is no column simultaneously covered by the first and the last $s$ locality-rows, after deleting the last $s$ locality-rows and the coordinates they cover, the resulting $H'$ has length $r+1\ge 5$, which also leads a contradiction. Therefore, $r\leq 3$.

If $s=3$, suppose there are two columns simultaneously covered by the first and the last three locality-rows, after deleting the first locality-row and other two locality-rows which intersect the first locality-row, we have $\gamma=3(r+1)-2$, a contradiction. Therefore, there exist at most one column simultaneously covered by the first and the last three locality-rows. After deleting the last three locality-rows and the columns they cover, the resulting $H'$ has length $r+1-1\leq 4$, which leads to $r\leq 4$.

If $s=2$, suppose there are three columns simultaneously covered by the first and the last two locality-rows, by pigeonhole principle, there exist two columns simultaneously covered by the first locality-row and one of the last two locality-rows. If we delete these two locality-rows and the columns they cover, we have $\gamma=s(r+1)-2$, a contradiction. Therefore, there are at most two columns covered by the first and last two locality-rows. After deleting the last two locality-rows and the columns they cover, the resulting $H'$ has length $r+1-2\leq 4$, which leads to $r\leq 5$.

To sum up, we obtain the following parameters.
%
%One can observe that there are at most $\max\{4-s,0\}$ the coordinates covered by the last $s$ locality-rows and first locality-row. Otherwise, the union of the support of the first locality-row and other $s-1$ locality-rows including the ones whose supports intersecting the first locality-row has size at most $s(r+1)-2$. Then after deleting these $s$ locality-rows and coordinates they cover, $\gamma\leq s(r+1)-2$, which leads a contradiction. Therefore, deleting the last $s$ locality-rows and the coordinates they cover, the resulting $H'$ has length $r+1-\max\{4-s,0\}\leq 4$, i.e. $r\leq 3+\max\{4-s,0\}$. Therefore, we obtain the following parameters.

If $r=2$, then $k=2s+1$, $n=k+\lc k/r\rc+1=3s+3$. Since $n'\leq r+1$, we have $n'$ can only be $3$. Thus, $H'$ is the parity-check matrix of a quaternary $[3,1,3]$ code. Due to the arbitrariness of the $s$ deleted locality-rows, all the locality-rows of $H$ have weight exactly $r+1$. The following $H$ gives the corresponding optimal code
\begin{flalign}\label{3s+3}
H=\left(\begin{array}{c}I_{s+1}\otimes(\begin{array}{ccc}1&1&1\end{array})\\ \hline \underbrace{(1~1~\ldots~1)}_{s+1}\otimes (\begin{array}{ccc}0&1&\omega\end{array}) \end{array}\right).
\end{flalign}

If $r=3$, then $k=3s+1$, $n=k+\lc k/r\rc+1=4s+3$. Since $(r+1)\nmid n$, $H'$ can be the parity-check matrix of a quaternary $[3,1,3]$ or $[4,2,3]$ code.
%When $s=1$, we have
%\begin{flalign}\label{4s+3_1}
%H=\left(\begin{array}{lllllll}1~&1~&1&\alpha~&0~&0~&0\\0&0&0&1&1&1&1\\ \hline0&\omega&\omega^2&1&0&\omega&\omega^2 \end{array}\right),
%\end{flalign}
%where $\alpha$ can be $0$ or $1$. When $s\ge 2$,
Then the following $H$ gives the corresponding optimal code
\begin{flalign}\label{4s+3_2}
H=\left(\begin{array}{c|c}\Big(\begin{array}{ccccccc}1&1&1&\alpha&0&0&0\\0&0&0&1&1&1&1\end{array}\Big)&\mathbf{0}_{2\times(4(s-1))}\\ \hline \mathbf{0}_{(s-1)\times 7}&I_{s-1}\otimes(\begin{array}{cccc}1&1&1&1\end{array})\\ \hline (\begin{array}{ccccccc}0&\omega&\omega^2&1&0&\omega&\omega^2\end{array})& \underbrace{(1~1~\ldots~1)}_{s-1}\otimes (\begin{array}{cccc}1&0&\omega&\omega^2\end{array}) \end{array}\right),
\end{flalign}
where $\alpha$ can be $0$ or $1$.
\begin{remark}
In the above case, after deleting the first $s$ locality-rows from $H$, the resulting $H'$ is the parity-check matrix of a quaternary $[4,2,3]$ MDS code, while after deleting the last $s$ locality-rows, the resulting $H'$ is the parity-check matrix of a quaternary $[3,1,3]$ MDS code.
\end{remark}
%, by the weight distribution of $H'$, the locality-rows of $H'$ intersect on at least 2 coordinates, therefore we must have $\lc k/r\rc-1=s=1$, and this row has weight $r$. Then $k=r+1$, $n=r+4$,

If $r=4$, $2\leq s\leq 3$, $k=4s+1$, $n=5s+3$. The following $Hs$ give the the corresponding optimal codes with parameters $[5s+3,4s+1,3]$ with $2\leq s\leq 3$ respectively.
\begin{flalign}
H=\left(\begin{array}{ccccccccccccc}
1&1&1&1&1&0&0&0&0&0&0&0&0\\
0&0&0&0&1&1&1&1&1&0&0&0&0\\
0&0&0&0&0&0&0&0&1&1&1&1&1\\
\hline
0&1&\omega&\omega^2&0&1&\omega&\omega^2&0&0&1&\omega&\omega^2
\end{array}
\right),~~
H=\left(\begin{array}{cccccccccccccccccc}
1&1&1&1&1&0&0&0&0&0&0&0&0&0&0&0&0&0\\
0&0&0&0&1&1&1&1&1&0&0&0&0&0&0&0&0&0\\
0&0&0&0&0&0&0&0&0&1&1&1&1&1&0&0&0&0\\
0&0&0&0&0&0&0&0&0&0&0&0&0&1&1&1&1&1\\
\hline
0&1&\omega&\omega^2&0&0&1&\omega&\omega^2&0&1&\omega&\omega^2&0&0&1&\omega&\omega^2
\end{array}
\right).\label{5s+3}
\end{flalign}

If $r=5$, then $s=2$, $k=5s+1$, $n=6s+3$. The following $H$ gives the corresponding optimal code with parameter $[6s+3,5s+1,3]$ with $s=2$.
\begin{flalign}
H=\left(\begin{array}{ccccccccccccccc}
1&1&1&1&1&1&0&0&0&0&0&0&0&0&0\\
0&0&0&0&0&1&1&1&1&1&1&0&0&0&0\\
0&0&0&0&1&0&0&0&0&0&1&1&1&1&1\\
\hline
0&1&\omega&\omega^2&0&1&0&1&\omega&\omega^2&1&0&1&\omega&\omega^2
\end{array}
\right).\label{6s+3}
\end{flalign}

\textbf{Case} $\mathbf{t=2}$: In this case, we have $n'=4$, i.e., $H'$ is the parity-check matrix of a quaternary $[4,2,3]$ MDS code.

If $s=1$, we have $k=r+2$, $n=r+5$ and $r\ge 3$. With the similar to the case when $t=1$ and $s=1$, we have $n\leq 20$ and $r\leq 15$.
%When $n=20$, we have  Moreover, to ensure the existence of $H'$ with such parameters, there are 4 zero coordinates in each locality-row of $H$. Applying Proposition \ref{prop_hamming} with $s=5$, we have $n\leq 21$ and $r\leq 16$. Moreover, one of locality-rows contains exact 4 zeros. If $n=21$, one can show that the code with parity-check matrix given in $(\ref{pcm_hamming_code})$ doesn't meet the restrictions above. Thus we have
\begin{flalign}\label{pcm_hamming_code_1_4}
H=\left(\begin{array}{cccccccccccccccccccccc}
1&1&1&1&1&1&1&1&1&1&1&1&1&1&1&1&0&0&0&0\\
1&1&1&\omega&\omega&\omega&\omega^2&\omega^2&\omega^2&\omega^2&\omega&1&0&0&0&0&1&1&1&1\\
\hline
1&\omega&\omega^2&1&\omega&\omega^2&1&\omega&\omega^2&0&0&0&\omega^2&0&\omega&1&\omega^2&\omega&1&0
\end{array}\right).
\end{flalign}
Therefore, we can puncture the first 0 to 12 columns from (\ref{pcm_hamming_code_1_4}) respectively to obtain an optimal LRC with parameters $[r+5,r+2,3]$, $3\leq r\leq 15$.

When $s\ge 2$, due to the arbitrariness of the deleted $s$ locality-rows during the construction of $H'$ and $\gamma=s(r+1)$, we have $r=3$, $k=3s+2$ and $n=4s+4$. The following parity-check matrix gives the corresponding optimal construction:
\begin{flalign}\label{4s+4}
H=\left(\begin{array}{c}I_{s+1}\otimes(\begin{array}{cccc}1&1&1&1\end{array})\\ \hline \underbrace{(1~1~\ldots~1)}_{s+1}\otimes (\begin{array}{cccc}0&1&\omega&\omega^2\end{array}) \end{array}\right).
\end{flalign}

\begin{itemize}
  \item For the case $l=\lc k/r\rc+1$
\end{itemize}

We have $n-k=l$, which implies that $H=H_1$.

If $r\mid k$, by Lemma \ref{r_mid_k}, the supports of all the locality-rows in $H$ must be pairwise disjoint. However, since $H'$ is a parity-check matrix of a quaternary $[5,3,3]$ or $[4,2,3]$ or $[3,1,3]$ MDS code, thus supports of the rows of $H'$ must intersect on some coordinates, this leads to a contradiction.

If $r\nmid k$, let $k=sr+t$, where $1\leq t\leq r-1$. With the same analysis as that for the case $l=\lc k/r\rc$, we can obtain a similar inequality as (\ref{inequality_1}):
\begin{flalign*}
k+\lc k/r \rc+1-(\lc k/r\rc-1)(r+1)\leq n-\gamma=n'\leq 5,
\end{flalign*}
this leads to $1\le t\le 3$.

Noted that $3\leq n'\leq 5$, for the sake of convenience, we divide our discussion into the following $3$ cases.

\textbf{Case} $\mathbf{n'}$ \textbf{can only be $5$}:

According to the weight distribution of the dual code of the quaternary $[5,3,3]$ MDS code corresponding to $H'$, we know that supports of the rows of $H'$ pairwise intersect on at least $3$ coordinates. Meanwhile, we have $\gamma=n-n'=s(r+1)-(3-t)$, which means there exist at most two columns covered by two members of the $s$ deleted locality-rows, therefore, we have $s=1$.

From the arbitrariness of the rows deleted to construct $H'$, the first and the last locality-rows of $H$ must intersect on at least $3$ coordinates which do not appear in the $H'$ obtained by deleting the first or the last locality-row of $H$. Since each row of $H'$ has weight at least $4$, the last row of $H$ has weight at least $7$. Thus $r\ge 6$.

When $t=1$ or $2$, the deleted locality-row of $H$ has weight $n-n'=r+t-2\leq r$. By the arbitrariness of the row deleted, each locality-row of $H$ has weight less than $r+1$, this contradicts the fact that $H$ has locality $r$.

When $t=3$, we have $\gamma=r+1$, $k=r+3$, $n=r+6$ and $d=3$. Applying Proposition \ref{prop_hamming} with $p=6$, we obtain $n=r+6\leq21$ and $6\leq r\leq 15$. Moreover, each of the locality-rows contains exact 5 zeros. When $n=21$, we have
\begin{flalign}\label{pcm_hamming_5}
H=\left(\begin{array}{ccccccccccccccccccccc}
1&1&1&1&1&1&1&1&1&1&1&1&1&1&1&1&0&0&0&0&0\\
1&1&1&\omega&\omega&\omega&\omega^2&\omega^2&\omega^2&\omega^2&\omega&0&0&0&0&1&0&1&1&1&1\\
1&\omega&\omega^2&1&\omega&\omega^2&1&\omega&\omega^2&0&0&\omega^2&0&\omega&1&0&1&\omega^2&\omega&1&0
\end{array}\right).
%H=\left(\begin{array}{ccccccccccccccccccccc}
%1&1&1&1&1&1&1&1&1&0&1&1&0&1&1&0&0&0&1&1&1\\
%1&1&1&\omega&\omega&\omega&\omega^2&\omega^2&\omega^2&0&\omega^2&\omega&1&0&0&1&1&1&0&0&1\\
%1&\omega&\omega^2&1&\omega&\omega^2&1&\omega&\omega^2&1&0&0&\omega^2&\omega^2&0&\omega&1&0&\omega&1&0
%\end{array}
%\right)
\end{flalign}
Therefore, we can puncture the first 0 to 9 columns from (\ref{pcm_hamming_5}) respectively to obtain an optimal LRC with parameters $[r+6,r+3,3]$, $6\leq r\leq 15$.

When $n'=3$ or $4$, through an analysis about the value of $\gamma$ and $n'$, we can obtain $s=1$, $t=1$. Since the optimal LRCs with such parameter have been discussed, we omit the discussion.
\subsection{$d=4$ and $H'$ contains three rows}

First, we need the following lemma to determine the nonexistence of optimal LRCs.

\begin{lem}\label{ban_struture}
If $H$ is a parity-check matrix of a linear code with minimum distance $d$, which has the following form,
\begin{flalign}\label{length_12_6_example}
H=\left(\begin{array}{ccccccccccc}
1&1&1~~~1&1&1&0&0&0~~~0&0&0&1\\
0&0&0~~~0&0&0&1&1&1~~~1&1&1&x_1\\
&&\multirow{1}{*}{\huge A}&&&&&\multirow{1}{*}{\huge B}&&&x_2\\
&&&&&&&&&&x_3
\end{array}\right),
\end{flalign}
%\begin{flalign}\label{length_12_6_example}
%H=\left(\begin{array}{ccccccccccc}
%1&1&1&1&1&1&0&0&0~~~0&0&0\\
%0&0&0&0&0&0&1&1&1~~~1&1&1\\
%0&1&0&\omega^2&\omega^2&1&&&\multirow{1}{*}{\huge A}&&\\
%0&0&1&\omega^2&1&\omega^2&&&&&
%\end{array}\right).
%\end{flalign}
where
\begin{flalign*}
\widetilde{A}:=\left(\begin{array}{ccccc}1&1&1~~~1&1&1\\ &&\multirow{1}{*}{\huge A}&&\\ &&&& \end{array}\right)~\text{and}~
\widetilde{B}:=\left(\begin{array}{ccccc}1&1&1~~~1&1&1\\ &&\multirow{1}{*}{\huge B}&&\\ &&&& \end{array}\right)
\end{flalign*}
are both parity-check matrices of a quaternary $[6,3,4]$ MDS code and $x_1\in\mathbb{F}_4^\ast$, $x_2,x_3\in\mathbb{F}_4$. Then, $d\leq 3$. %Then, we can't add any other columns into $H$ to obtain a parity-check matrix of a code containing $H$ with length $n\ge 13$.
\end{lem}
\begin{proof}[Proof of Lemma \ref{ban_struture}]
The idea of the proof is to regard the columns of $H$ in (\ref{length_12_6_example}) as the point of $PG(1,4)$, then try to analysis the linear relationship between them. We will find that if $H$ has the form as (\ref{length_12_6_example}), there exist $4$ columns linearly dependent. We finish the proof.

We use $a_i$, $b_i$ to denote the $i_{th}$ column of $\widetilde{A}$, $\widetilde{B}$ respectively. Noted that the first coordinate of $a_i$s are all 1, therefore, for distinct $i,j\in[6]$, $a_i+a_j$ can be regarded as a point $p_{a_i+a_j}$ of $PG(1,4)$. Based on this observation, we have the following claim demonstrating the linear dependency of columns in $A$.

\begin{claim}
If there exist $(a_{i_1},a_{i_2})$ and $(a_{j_1},a_{j_2})$, for some $i_1,i_2,j_1,j_2\in[6]$, such that $p_{a_{i_1}+a_{i_2}}=p_{a_{j_1}+a_{j_2}}$, then $\{i_1,i_2\}\cap \{j_1,j_2\}=\emptyset$. Moreover, for each point in $PG(1,4)$, there exist exactly three $2$-sets $\{a_i,a_j\}$ corresponding to a same point $p_{a_i+a_j}$.
\end{claim}
\begin{proof}[Proof of the claim]
Assume that $\{i_1,i_2\}\cap \{j_1,j_2\}\neq\emptyset$, then one can easily find $3$ columns in multiset $\{a_{i_1},a_{i_2},a_{j_1},a_{j_2}\}$ that are linearly dependent. This contradicts the requirement on the minimum distance of $\widetilde{A}$. Therefore, we know that there exist at most $3$ distinct $2$-sets $\{i,j\}\subseteq [6]$ such that $a_i+a_j$ corresponds to the same point in $PG(1,4)$. Besides, since there are ${6\choose 2}=15$ 2-sets and $(4^2-1)/(4-1)=5$ points in $PG(1,4)$ in all. For a fixed point $p_0$ in $PG(1,4)$, if there are less than three $2$-sets $\{i,j\}\subseteq [6]$ such that $p_{a_i+a_j}=p_0$, by pigeonhole principle, there exist at least four $2$-sets ${i,j}$ such that every $a_i+a_j$ corresponds to a same point, which contradicts the requirement on $\widetilde{A}$ to have minimum distance $4$.
\end{proof}
Thus, w.l.o.g., we assume that the correspondence between the columns of $\widetilde{A}$ and the points in $PG(1,4)$ has the following form:
\begin{flalign*}
p_{a_1+a_2}=p_{a_3+a_6}=p_{a_4+a_5}=c_1,\\
p_{a_1+a_3}=p_{a_2+a_5}=p_{a_4+a_6}=c_2,\\
p_{a_1+a_4}=p_{a_5+a_6}=p_{a_2+a_3}=c_3,\\
p_{a_1+a_5}=p_{a_3+a_4}=p_{a_2+a_6}=c_4,\\
p_{a_1+a_6}=p_{a_2+a_4}=p_{a_3+a_5}=c_5,
\end{flalign*}
where $c_1,\ldots,c_5$ are different points in $PG(1,4)$.

Moreover, if $a_{i_1}+a_{i_2}=a_{j_1}+a_{j_2}$ for some $i_1,i_2,j_1,j_2\in [6]$, then $a_{i_1}+a_{j_1}=a_{i_2}+a_{j_2}$. Besides, from the correspondence above, we can obtain another linear equation $a_{i_1}+a_{j_1}=\lambda_0 (a_{i_2}+a_{j_3})$ for some $j_3\in[6]\setminus\{i_1,i_2,j_1,j_2\}$ and $\lambda_0\in \mathbb{F}_4^*$. This leads to $3$ columns of $a_{i}$s being linearly dependent, which contradicts the fact that $d=4$. Therefore, w.l.o.g., we can also assume that
\begin{flalign*}
a_1+a_2=\omega (a_3+a_6)=\omega^2( a_4+a_5)\sim c_1,\\
a_1+a_3=\omega (a_2+a_5)=\omega^2( a_4+a_6)\sim c_2,\\
a_1+a_4=\omega (a_5+a_6)=\omega^2( a_2+a_3)\sim c_3,\\
a_1+a_5=\omega (a_3+a_4)=\omega^2( a_2+a_6)\sim c_4,\\
a_1+a_6=\omega (a_2+a_4)=\omega^2( a_3+a_5)\sim c_5,
\end{flalign*}
where $\sim$ stands for the aforementioned correspondence.

By the uniqueness of the quaternary $[6,3,4]$ MDS code, we can assume that
\begin{flalign*}
\widetilde{A}=\left(\begin{array}{ccccll}1&1&1&1&1&1\\0&1&0&\omega^2&1&\omega^2\\0&0&1&\omega^2&\omega^2&1 \end{array}\right).
\end{flalign*}
Thus we have
\begin{flalign}\label{correspondence1}
a_1+a_2&=\omega (a_3+a_6)=\omega^2( a_4+a_5)\sim (1,0), \nonumber\\
a_1+a_3&=\omega (a_2+a_5)=\omega^2( a_4+a_6)\sim (0,1),\nonumber\\
a_1+a_4&=\omega (a_5+a_6)=\omega^2( a_2+a_3)\sim (1,1),\\
a_1+a_5&=\omega (a_3+a_4)=\omega^2( a_2+a_6)\sim (1,\omega^2),\nonumber\\
a_1+a_6&=\omega (a_2+a_4)=\omega^2( a_3+a_5)\sim (1,\omega)\nonumber.
\end{flalign}

Let $\widetilde{B}$ be an arbitrary parity-check matrix of a quaternary $[6,3,4]$ code, w.l.o.g., assume the correspondence between the columns of $\widetilde{B}$ and the points in $PG(1,4)$ has the following form:
\begin{flalign}\label{correspondence2}
b_1+b_2&=\omega (b_3+b_6)=\omega^2( b_4+b_5)\sim d_1,\nonumber\\
b_1+b_3&=\omega (b_2+b_5)=\omega^2( b_4+b_6)\sim d_2,\nonumber\\
b_1+b_4&=\omega (b_5+b_6)=\omega^2( b_2+b_3)\sim d_3,\\
b_1+b_5&=\omega (b_3+b_4)=\omega^2( b_2+b_6)\sim d_4,\nonumber\\
b_1+b_6&=\omega (b_2+b_4)=\omega^2( b_3+b_5)\sim d_5,\nonumber
\end{flalign}
where $d_1,\ldots,d_5$ are different points in $PG(1,4)$.

For each $i\in [6]$, denote $\alpha_i$ as the first $6$ columns in $H$ given by $(\ref{length_12_6_example})$ and $\beta_i$ as the $7_{th}$ to $12_{th}$ columns of $H$. Then the linear dependencies among $\alpha_i$s and $\beta_i$s are the same as those among $a_i$s and $b_i$s.

Now, we show that the last column of $H$ in (\ref{length_12_6_example}) must be linear combinations of some $\alpha_{i},\beta_{j}$, $i,j\in [6]$, which contradicts that $d=4$.

First, consider the columns of the form $\alpha_i+\lambda \beta_j$, where $i,j\in [6]$ and $\lambda\in\mathbb{F}_4^\ast$. If there exist two columns of different form $\alpha_{i_1}+\lambda_1 \beta_{i_2}$ and $\alpha_{j_1}+\lambda_2 \beta_{j_2}$, such that $\alpha_{i_1}+\lambda_1 \beta_{i_2}=\alpha_{j_1}+\lambda_2 \beta_{j_2}$, then we have $\lambda_1=\lambda_2$ and this leads to a linear relation between $\alpha_{i_1},\beta_{j_1},\alpha_{i_2},\beta_{j_2}$.

Since $\alpha_i$s and $\beta_i$s inherit the linear dependencies among $a_i$s and $b_i$s, by (\ref{correspondence1}) and (\ref{correspondence2}), there exists $i_1\in[6]$ such that there is a unique $i_j\in[6]$ for each $2\leq j\leq 6$, such that $\alpha_1+\alpha_j=\beta_{i_1}+\beta_{i_j}$. Then combining this linear equation with the linear relations among columns of $\widetilde{B}$, we can obtain the following classification of all the columns of the form $\alpha_i+\lambda \beta_j$:
\begin{enumerate}
  \item $\alpha_1+\beta_{i_1}=\alpha_2+\beta_{i_2}=\cdots=\alpha_6+\beta_{i_6}$, thus columns of the form $\alpha_j+\beta_{i_j}$ $(j\in[6])$ provide only 1 element in $\mathbb{F}_4^4$.
  \item $\alpha_m+\beta_{i_n}=\alpha_n+\beta_{i_m}$, when $m\ne n\in[6]$. Thus columns of the form $\alpha_m+\beta_{i_n}$ provide ${6\choose 2}=15$ different elements in $\mathbb{F}_4^4$.
  \item For $\lambda\in\{\omega,\omega^2\}$ and each $s\in[6]$, $\alpha_s+\lambda \beta_{i_s}$ is unique. Thus columns of this form provide $6\times 2=12$ different elements in $\mathbb{F}_4^4$.
  \item For $\lambda\in\{\omega,\omega^2\}$ and $\{s,t\}\subseteq[6]$, by (\ref{correspondence1}) and (\ref{correspondence2}), there are exactly two other $\{s_1,t_1\},\{s_2,t_2\}\subseteq [6]$ such that $\alpha_{s}+\lambda \beta_{i_{t}}=\alpha_{s_1}+\lambda \beta_{i_{t_1}}=\alpha_{s_2}+\lambda \beta_{i_{t_2}}$. Therefore, columns of this form provide ${(6\times 5\times2)}/{3}=20$ different elements in $\mathbb{F}_4^4$.
\end{enumerate}

To sum up, there are 48 different columns of the form $\alpha_i+\lambda \beta_j$, where $i,j\in [6]$ and $\lambda\in\mathbb{F}_4^\ast$. On the other hand, since $x_2\in\mathbb{F}_4^\ast$ and $x_3,x_4\in\mathbb{F}_4$, we only have $3\times 4\times 4=48$ different choices for $(1,x_2,x_3,x_4)^{T}$. Therefore, all columns of this form are linear combinations of $\alpha_i$ and $\beta_j$ and we finish the proof.
\end{proof}

The discussion in this subsection is similar to those in Subsections \ref{subsection_4_1} and \ref{d=3}, in order to avoid the redundancy, we shall omit some of the repeated details.

In this subsection, $H'$ can be the parity-check matrix of a quaternary $[6,3,4]$ or $[5,2,4]$ or $[4,1,4]$ MDS code. And we have $n-k-\lc k/r\rc=d-2=2$. Since $H_1$ has $l\ge\lc k/r\rc$ rows, by (\ref{inequality}), we have
\begin{flalign*}
\lc k/r\rc\leq l\leq n-k=\lc k/r\rc+2.
\end{flalign*}

\begin{itemize}
  \item For the case $l=\lc k/r\rc$
\end{itemize}

Then $n-k=l+2$, which implies that $H_2$ contains two rows. According to the construction procedure of $H'$, $H'$ contains a locality-row covering all the coordinates remained after the deletion. Therefore, $n'\leq r+1$ and the linear code $\mc{C}'$ with parity-check matrix $H'$ has weight distribution with $A_{n'}>0$. By Theorem \ref{MDS_weight_distribution}, the dual codes of quaternary $[6,3,4]$, [5,2,4] and $[4,1,4]$ MDS codes all satisfy this condition. Thus $H'$ can be the parity-check matrix with all these $3$ kinds of parameters. This leads to $4\leq n'\leq 6$ and $r\geq 3$.

If $r\mid k$, we set $k=sr$ for some $s>0$. By Lemma \ref{r_mid_k}, we have $(r+1)\mid n$. We also have $n=k+l+2=s(r+1)+2$. This contradicts the fact that $(r+1)\mid n$. Therefore, optimal LRCs with such parameters do not exist.

If $r\nmid k$, let $k=sr+t$, where $1\le t\le r-1$. Recall that $\gamma$ is the number of columns covered by the supports of the deleted $\lc k/r\rc-1$ locality-rows. Thus, $\gamma\leq(\lc k/r\rc-1)(r+1)$. Then
\begin{flalign}\label{inequality_under_d4}
k+\lc k/r \rc+2-(\lc k/r\rc-1)(r+1)\leq n-\gamma=n'\leq 6,
\end{flalign}
this leads to $1\leq t\leq 3$.

\textbf{Case} $\mathbf{t=1}$:  In this case, we have $4\le n'\le 6$ and $H'$ can be the parity-check matrix of a quaternary $[6,3,4]$ or $[5,2,4]$ or $[4,1,4]$ MDS code.

%which implies that the supports of the $s$ removed locality-rows are disjoint and each has support size exactly $r+1$. Similarly, we can show that when $n'=5$ or $n'=6$, the supports of $s$ removed locality-rows intersect on one coordinates or two coordinates respectively. First, we consider additional circumstances if $s=1$, $n'=4$ as we do in last subsection.\\
If $s=1$, we have $k=r+1$ and $n=r+5$. It can be observed that after we remove arbitrary $1$ locality-row with $m$ zeros and columns it covered, we obtain a parity-check matrix of code with length $m$. Therefore each locality-row of $H$ must have $4$, $5$ or $6$ zeros to ensure the existence of $H'$. Moreover, one of locality-rows must have exactly 4 zeros to meet the requirement for parameters. Since $d=4$, any 3 columns of $H$ are linearly independent. By Lemma \ref{finite_geometry}.2, $n=r+5\leq m_2(3,4)=17$ and $r\leq 12$. Take each column of $H$ as a point of $PG(3,4)$, when $n=17$, they form an ovoid in $PG(3,4)$. By Lemma \ref{property_ovoid}, there is $1$ or $5$ zeros in each row of $H$, which doesn't meet the restrictions above. Thus, $n\leq 16$ and $r\leq 11$. When $n=16$, we have
\begin{flalign}\label{length_16}
H=\left(\begin{array}{cccccccccccccccc}
1&1&1&1&0&0&0&0&\omega&1&\omega^2&\omega&\omega&\omega^2&1&1\\
0&0&0&0&1&1&1&1&\omega&\omega^2&\omega^2&\omega&\omega^2&\omega^2&0&1\\
\hline
0&1&0&\omega^2&0&1&0&\omega^2&\omega&\omega&\omega^2&0&1&1&\omega^2&\omega\\
0&0&1&\omega^2&0&0&1&\omega^2&1&1&1&1&\omega^2&1&1&0
\end{array}\right).
\end{flalign}
Therefore, we can puncture the last 0 to 8 columns from $(\ref{length_16})$ respectively to obtain an optimal LRC with parameters $[r+5,r+1,4]$, $3\leq r\leq 11$.

When $s\geq 2$ and $n'=4$, we have $\gamma=n-n'=s(r+1)$. Thus, the supports of the $s$ deleted locality-rows are pairwise disjoint with size $r+1$.
%Due to the arbitrariness of these $s$ locality-rows, all the locality-rows of $H$ must have weight exactly $r+1$.
Similarly, when $n'=5$ or $6$, one can show that the support of $s$ deleted locality-rows intersect in one column or two columns respectively.

The following discussion is similar to the case $l=\lceil k/r\rceil$ and $s\ge 2$ in Subsection \ref{d=3}, so we omit some details. W.l.o.g., suppose that the weight of the first locality-row of $H$ is $r+1$. When $s\ge 4$, by the arbitrariness of these $s$ rows, one can show that $r\leq 5$.

If $s=3$, there are at most two columns covered by the first locality-row and last three locality-rows from the similar discussion. After deleting the last three locality-rows and the columns they cover, the resulting $H'$ has length $r+1-2\leq 6$, which leads to $r\leq 7$. When $r=7$, this deletion will lead to an $H'$ of length $6$. If the first two locality-rows intersect on two columns, then after deleting these two locality-rows together with any other locality-row, according to $\gamma$, all these $3$ locality-rows have weight $r+1$ and the resulting $H'$ has length $n'=6$. Thus, by Lemma \ref{ban_struture}, to avoid submatrix of form (\ref{length_12_6_example}), there exists another locality-row (the $i_{th}$ locality-row, $(3\leq i\leq 4)$) intersecting the first two. Therefore, if we delete the first two and the $i_{th}$ locality-rows and the columns they cover, the corresponding $\gamma\leq 3(r+1)-3$, a contradiction. Similarly, if the first two locality-rows intersect on one column, one can also obtain a contradiction. Therefore, we have $r\leq6$.

If $s=2$, there are at most four columns covered by the first locality-row and last two locality-rows. Delete the last two locality-rows and the columns they cover, the resulting $H'$ has length $r+1-4\leq 6$, this leads to $r\leq 9$. When $8\leq r\leq 9$, since $n'\leq 6$, w.l.o.g., assume that the second locality-row of $H$ intersects the first locality-row of $H$ on two columns and the third locality-row intersects the first locality-row on $r-7$ columns. By the definition of $\gamma$, after deleting the first two locality-rows and the columns they cover, we have $n'=6$. Therefore, the last two locality-rows must intersect in some column which is not covered by the first locality-row. This indicates that $H$ contains a submatrix of form $(\ref{length_12_6_example})$ in Lemma \ref{ban_struture}, a contradiction. Thus, $r\leq 7$.

To sum up, we obtain the following parameters.

If $r=3$, we have $k=3s+1$, $n=k+\lc k/r\rc+2=4s+4$ and $n'=4$. Thus, $H'$ is the parity-check matrix of a quaternary $[4,1,4]$ code. Due to the arbitrariness of the $s$ deleted locality-rows, all the locality-rows of $H$ have weight exactly $r+1$. The following $H$ gives the corresponding optimal LRC
\begin{flalign}
H=\left(\begin{array}{c}I_{s+1}\otimes(\begin{array}{cccc}1&1&1&1\end{array})\\ \hline \underbrace{(1~1~\ldots~1)}_{s+1}\otimes \Big(\begin{array}{cccc}0&1&0&1\\0&0&1&1\end{array}\Big) \end{array}\right).\label{4s+4_3s+1}
\end{flalign}

If $r=4$, we have $k=4s+1$ and $n=5s+4$. Since $(r+1)\nmid n$, $H'$ can be the parity-check matrix of a quaternary $[4,1,4]$ or $[5,2,4]$ MDS code.
%When $s=1$, we have
%\begin{flalign}\label{5s+4_1}
%H=\left(\begin{array}{llllllllll}
%1~&1~&1~&1&\alpha&0&0~&0~&0\\
%0&0&0&0&1&1&1&1&1\\
%\hline
%0&1&0&\omega^2&\omega^2&\omega^2&0&1&0\\
%0&0&1&\omega^2&1&\omega^2&1&0&0
% \end{array}\right).
%\end{flalign}
%where $\alpha$ can be $0$ or $1$. When $s\ge 2$,
Then the following $H$ gives the corresponding optimal code.
\begin{flalign}
H=\left(\begin{array}{c|c}\Big(\begin{array}{ccccccccc}1&1&1&1&\alpha&0&0&0&0\\0&0&0&0&1&1&1&1&1\end{array}\Big)&\mathbf{0}_{2\times(5(s-1))}\\ \hline \mathbf{0}_{(s-1)\times 9}&I_{s-1}\otimes(\begin{array}{ccccc}1&1&1&1&1\end{array})\\ \hline \Big(\begin{array}{cccclcccc} 0&1&0&\omega^2&\omega^2&\omega^2&0&1&0\\
0&0&1&\omega^2&1&\omega^2&1&0&0\end{array}\Big)& \underbrace{(1~1~\ldots~1)}_{s-1}\otimes \Big(\begin{array}{ccccl}0&1&0&\omega^2&\omega^2\\0&0&1&\omega^2&1\end{array}\Big) \end{array}\right),\label{5s+4_4s+1}
\end{flalign}
where $\alpha$ can be $0$ or $1$.
\begin{remark}
In the above case, after deleting the first $s$ locality-rows from $H$, the resulting $H'$ is a parity-check matrix of a quaternary $[5,2,4]$ MDS code, while after deleting the last $s$ locality-rows, the resulting $H'$ is a parity-check matrix of a quaternary $[4,1,4]$ MDS code.
\end{remark}

If $r=5$, we have $k=5s+1$ and $n=6s+4$. Since $(r+1)\nmid n$, $H'$ can be the parity-check matrix of a quaternary $[4,1,4]$ or $[5,2,4]$ or $[6,3,4]$ MDS code.
%When $s=1$, we have
%\begin{flalign}\label{6s+4_1}
%H=\left(\begin{array}{llllllllll}
%1~&1~&1~&1&\alpha&\beta&0&0~&0~&0\\
%0&0&0&0&1&1&1&1&1&1\\
%\hline
%0&1&0&\omega^2&\omega^2&1&\omega^2&0&1&0\\
%0&0&1&\omega^2&1&\omega^2&\omega^2&1&0&0
% \end{array}\right),
%\end{flalign}
%where $(\alpha,\beta)$ can be $(0,0)$ or $(1,1)$. When $s\ge 2$,
The following $H$ gives the corresponding optimal code.
\begin{flalign}
H=\left(\begin{array}{c|c}\Big(\begin{array}{cccccccccc}1&1&1&1&\alpha&\beta&0&0&0&0\\0&0&0&0&1&1&1&1&1&1\end{array}\Big)&\mathbf{0}_{2\times(6(s-1))}\\ \hline \mathbf{0}_{(s-1)\times 10}&I_{s-1}\otimes(\begin{array}{cccccc}1&1&1&1&1&1\end{array})\\ \hline \Big(\begin{array}{ccccllccccc}0&1&0&\omega^2&\omega^2&1&\omega^2&0&1&0\\
0&0&1&\omega^2&1&\omega^2&\omega^2&1&0&0\end{array}\Big)& \underbrace{(1~1~\ldots~1)}_{s-1}\otimes \Big(\begin{array}{ccccll}0&1&0&\omega^2&1&\omega^2\\0&0&1&\omega^2&\omega^2&1\end{array}\Big) \end{array}\right),\label{6s+4_5s+1_1}
\end{flalign}
where $(\alpha,\beta)$ can be either $(0,0)$, $(0,1)$, $(1,0)$ or $(1,1)$.
\begin{remark}
In the above case, after deleting the first $s$ locality-rows from $H$, the resulting $H'$ is a parity-check matrix of a quaternary $[6,3,4]$ MDS code, while after deleting the last $s$ locality-rows, the resulting $H'$ corresponds to a parity-check matrix of a quaternary $[4,1,4]$ MDS code.
\end{remark}

Moreover, the following $H$ can provide another construction with the same parameter. When $s=2$, we have
\begin{flalign}\label{6s+4_5s+1_2}
H=\left(\begin{array}{llllllllllllllll}
1~~&1~~&1~~&1&1&\alpha&0&0&0~~&0~~&0~~&0~~&0~~&0&0&0\\
0&0&0&0&0&1&1&1&1&1&\beta&0&0&0&0&0\\
0&0&0&0&0&0&0&0&0&0&1&1&1&1&1&1\\
\hline
0&1&0&\omega^2&\omega^2&1&\omega^2&\omega^2&0&1&1&1&\omega&\omega&\omega^2&\omega^2\\
0&0&1&\omega^2&1&\omega^2&1&\omega^2&1&0&0&\omega^2&0&\omega&\omega^2&\omega
 \end{array}\right),
\end{flalign}
where $\alpha$, $\beta$ can be $0$ or $1$. And when $s\ge 3$, based on (\ref{6s+4_5s+1_2}), we have
\begin{flalign}\label{6s+4_5s+1_3}
H=\left(\begin{array}{c|c}\mathbf{A}&\mathbf{0}_{3\times(6(s-2))}\\
\hline \mathbf{0}_{(s-2)\times 16}&I_{s-2}\otimes(\begin{array}{cccccc}1&1&1&1&1&1\end{array})\\ \hline
\mathbf{B}& \underbrace{(1~1~\ldots~1)}_{s-2}\otimes \Big(\begin{array}{ccccll}0&1&0&\omega^2&1&\omega^2\\0&0&1&\omega^2&\omega^2&1\end{array}\Big) \end{array}\right),
\end{flalign}
where
\begin{flalign*}
A=\left(\begin{array}{cccccccccccccccc}
1&1&1&1&1&\alpha&0&0&0&0&0&0&0&0&0&0\\
0&0&0&0&0&1&1&1&1&1&\beta&0&0&0&0&0\\
0&0&0&0&0&0&0&0&0&0&1&1&1&1&1&1\end{array}\right)
\end{flalign*}
$\alpha$, $\beta$ can be 0 or 1, and
\begin{flalign*}
B=\left(\begin{array}{cccccccccccccccc}
0&1&0&\omega^2&\omega^2&1&\omega^2&\omega^2&0&1&1&1&\omega&\omega&\omega^2&\omega^2\\
0&0&1&\omega^2&1&\omega^2&1&\omega^2&1&0&0&\omega^2&0&\omega&\omega^2&\omega\end{array}\right).
\end{flalign*}

\begin{remark}
In the above case, after deleting the first $s$ locality-rows from $H$, the resulting $H'$ is a parity-check matrix of a quaternary $[6,3,4]$ MDS code, while after deleting the last $s$ locality-rows, or the first row and last $s-1$ locality-rows, the resulting $H'$ is a parity-check matrix of a quaternary $[5,2,4]$ MDS code.
\end{remark}

If $r=6$, we have $k=6s+1$ and $n=7s+4$, $2\leq s\leq 3$. The following $H$s give the parity-check matrices of the corresponding optimal codes with parameters $[7s+4,6s+1,4]$ with $s=2,3$ respectively.

\begin{flalign*}
H=\left(\begin{array}{cccccccccccccccccc}
1&1&1&1&1&1&1&0&0&0&0&0&0&0&0&0&0&0\\
0&0&0&0&0&0&1&1&1&1&1&1&1&0&0&0&0&0\\
0&0&0&0&0&0&0&0&0&0&0&1&1&1&1&1&1&1\\
\hline
0&1&0&\omega^2&\omega^2&1&1&1&1&\omega&\omega^2&0&\omega^2&0&1&0&\omega^2&\omega^2\\
0&0&1&\omega^2&1&\omega^2&0&\omega^2&\omega&0&0&0&\omega^2&0&0&1&\omega^2&1
\end{array}
\right),
\end{flalign*}

\begin{flalign}\label{7s+4}
H=\left(\begin{array}{ccccccccccccccccccccccccc}
1&1&1&1&1&1&1&0&0&0&0&0&0&0&0&0&0&0&0&0&0&0&0&0&0\\
0&0&0&0&0&0&1&1&1&1&1&1&1&0&0&0&0&0&0&0&0&0&0&0&0\\
0&0&0&0&0&1&0&0&0&0&0&0&0&1&1&1&1&1&1&0&0&0&0&0&0\\
0&0&0&0&0&0&0&0&0&0&0&0&0&0&0&0&0&0&1&1&1&1&1&1&1\\
\hline
1&1&\omega&\omega&0&0&0&0&1&0&\omega^2&\omega^2&1&1&1&\omega&\omega&0&0&0&1&0&\omega^2&\omega^2&1\\
0&\omega^2&0&1&\omega^2&\omega&1&0&0&1&\omega^2&1&\omega^2&0&\omega^2&0&1&\omega^2&1&0&0&1&\omega^2&1&\omega^2
\end{array}
\right).
\end{flalign}

If $r=7$, we have $k=7s+1$ and $n=8s+4$, $s=2$. The following $H$ gives the parity-check matrix of the corresponding optimal codes with parameter $[8s+4,7s+1,4]$ with $s=2$.

\begin{flalign}\label{8s+4}
H=\left(\begin{array}{cccccccccccccccccccc}
1&1&1&1&1&1&1&1&0&0&0&0&0&0&0&0&0&0&0&0\\
0&0&0&0&0&0&1&1&1&1&1&1&1&1&0&0&0&0&0&0\\
0&0&0&0&0&0&0&0&0&0&0&0&1&1&1&1&1&1&1&1\\
\hline
0&1&0&\omega^2&\omega^2&1&1&1&1&1&\omega^2&\omega^2&1&\omega^2&0&1&0&\omega^2&\omega^2&1\\
0&0&1&\omega^2&1&\omega^2&0&1&\omega^2&\omega&0&1&0&\omega^2&0&0&1&\omega^2&1&\omega^2
\end{array}
\right).
\end{flalign}

\textbf{Case} $\mathbf{t=2}$: In this case, $5\le n'\le 6$ and $H'$ can be the parity-check matrix of a quaternary $[5,2,4]$ or $[6,3,4]$ MDS code.

If $s=1$, we have $k=r+2$, $n=r+6$ and $r\ge 4$. There are $5$ or $6$ zeros in each locality-row of $H$ to ensure the existence of $H'$ with such parameters. Moreover, one of them has exactly 5 zeros to meet the requirement for parameters. Since $d=4$, we know that any $3$ columns of $H$ are linearly independent. Thus, we have $n=r+6\leq m_2(3,4)=17$ and $r\leq 11$. %Moreover, take each column vector of $H$ as a point of $PG(3,4)$, they actually form an elliptic quadric.
When $n=17$, we have
\begin{flalign}\label{length_17}
H=\left(\begin{array}{ccccccccccccccccc}
1&1&1&1&1&0&0&0&0&0&\omega&1&\omega^2&\omega&\omega&\omega^2&1\\
0&0&0&0&0&1&1&1&1&1&\omega&\omega^2&\omega^2&\omega&\omega^2&\omega^2&1\\
\hline
0&1&0&\omega^2&\omega^2&0&1&0&\omega^2&\omega^2&\omega&\omega&\omega^2&0&1&1&\omega\\
0&0&1&\omega^2&1&0&0&1&\omega^2&1&1&1&1&1&\omega^2&1&0
\end{array}\right).
\end{flalign}
Therefore, we can puncture the last 0 to 7 columns from $(\ref{length_17})$ respectively to obtain an optimal LRC with parameters $[r+6,r+2,4]$, $4\leq r\leq11$.

When $s\geq 2$ and $n'=5$, we have $\gamma=n-n'=s(r+1)$. Thus, the supports of the $s$ deleted locality-rows are pairwise disjoint with size $r+1$.
%Due to the arbitrariness of these $s$ locality-rows, all the locality-rows of $H$ must have weight exactly $r+1$.
Similarly, when $n'=6$, one can show that the support of $s$ deleted locality-rows intersect on one column.

The following discussion is similar to the case $l=\lfloor k/r\rfloor$ and $s\ge 2$ in Subsection \ref{d=3}. W.l.o.g., suppose that the weight of the first locality-row of $H$ is $r+1$. When $s\ge 4$, by the arbitrariness of the $s$ locality-rows, one can show that $r\leq 5$.

If $s=3$, there exist at most one column simultaneously covered by the first and last three locality-rows by similar discussion. After deleting the last three locality-rows and the columns they cover, the resulting $H'$ has length $r+1-1\leq 6$, which leads to $r\leq 6$. When $r=6$, the same as case $t=1$ and $s=3$, optimal LRCs with such parameters do not exist.

If $s=2$, there exist at most two columns simultaneously covered by the first and last two locality-rows by similar discussion. After deleting the last two locality-rows and the columns they cover, the resulting $H'$ has length $r+1-2\leq 6$, which leads to $r\leq 7$. When $r=7$, since $n'\leq 6$, by the definition of $\gamma$, the second and third locality-rows both intersect the first locality-row in different column. If we delete the first two locality-rows (or the first and the third locality-rows), the resulting $H'$ has length $n'=6$, which indicates that $H$ contains a submatrix of form $(\ref{length_12_6_example})$ in Lemma \ref{ban_struture}, a contradiction. Thus, $r\leq6$.

%When $s=3$ and $r=6$, with similar arguments as those for the case $t=1$ and $s=3$, one can also prove that optimal LRCs for this case do not exist.

Now we give the explicit construction respectively.

If $r=4$, then $k=4s+2$ and $n=5s+5$. Thus, $H'$ is the parity-check matrix of a quaternary $[5,2,4]$ MDS code. The following $H$ gives the corresponding optimal code
\begin{flalign}\label{5s+5_4s+2}
H=\left(\begin{array}{c}I_{s+1}\otimes(\begin{array}{ccccc}1&1&1&1&1\end{array})\\ \hline \underbrace{(1~1~\ldots~1)}_{s+1}\otimes \Big(\begin{array}{ccccc}0&1&0&\omega^2&\omega^2\\0&0&1&\omega^2&1\end{array}\Big) \end{array}\right).
\end{flalign}

If $r=5$, then $k=5s+2$ and $n=6s+5$. Since $(r+1)\nmid n$, $H'$ can be the parity-check matrix of a quaternary $[5,2,4]$ or $[6,3,4]$ MDS code.

Then the following $H$ gives the corresponding optimal code
\begin{flalign}\label{6s+5_5s+2}
H=\left(\begin{array}{c|c}\Big(\begin{array}{ccccccccccc}1&1&1&1&1&\alpha&0&0&0&0&0\\0&0&0&0&0&1&1&1&1&1&1\end{array}\Big)&\mathbf{0}_{2\times(6(s-1))}\\ \hline \mathbf{0}_{(s-1)\times 11}&I_{s-1}\otimes(\begin{array}{cccccc}1&1&1&1&1&1\end{array})\\ \hline \Big(\begin{array}{cccclllcccc}0&1&0&\omega^2&\omega^2&1&\omega^2&\omega^2&0&1&0\\0&0&1&\omega^2&1&\omega^2&1&\omega^2&1&0&0 \end{array}\Big)& \underbrace{(1~1~\ldots~1)}_{s-1}\otimes \Big(\begin{array}{cccccc}0&1&0&1&\omega^2&\omega^2\\0&0&1&\omega^2&\omega^2&1\end{array}\Big) \end{array}\right),
\end{flalign}
where $\alpha$ can be $0$ or $1$.
\begin{remark}
In the above case, after deleting the first $s$ locality-rows from $H$, the resulting $H'$ is a parity-check matrix of a quaternary $[6,3,4]$ MDS code, while after deleting the last $s$ locality-rows, the resulting $H'$ is a parity-check matrix of a quaternary  $[5,2,4]$ MDS code.
\end{remark}

If $r=6$, $k=6s+2$, $n=7s+5$ with $s=2$. The following $H$ gives the parity-check matrix of the corresponding optimal codes with parameter $[7s+5,6s+2,4]$ with $s=2$.

\begin{flalign}\label{7s+5}
H=\left(\begin{array}{ccccccccccccccccccc}
1&1&1&1&1&1&1&0&0&0&0&0&0&0&0&0&0&0&0\\
0&0&0&0&0&0&1&1&1&1&1&1&1&0&0&0&0&0&0\\
0&0&0&0&0&0&0&0&0&0&0&0&1&1&1&1&1&1&1\\
\hline
0&1&0&\omega^2&\omega^2&1&\omega^2&1&\omega^2&\omega^2&\omega&1&\omega^2&0&1&0&\omega^2&\omega^2&1\\
0&0&1&\omega^2&1&\omega^2&1&\omega&\omega^2&\omega&\omega^2&1&1&0&0&1&\omega^2&1&\omega^2
\end{array}
\right).
\end{flalign}

\textbf{Case} $\mathbf{t=3}$: When $s=1$, $k=r+3$, $n=r+7$. Since $n'=6$ in this case, we have $r\ge 5$. For $r=5$, $k=8$ and $n=12$, the following $H$ gives the corresponding optimal code
\begin{flalign}\label{length_12_6_only}
H=\left(\begin{array}{cccccccccccc}
1&1&1&1&1&1&0&0&0&0&0&0\\
0&0&0&0&0&0&1&1&1&1&1&1\\
\hline
0&1&0&\omega^2&\omega^2&1&0&1&0&\omega^2&\omega^2&1\\
0&0&1&\omega^2&1&\omega^2&0&0&1&\omega^2&1&\omega^2
\end{array}\right).
\end{flalign}

Next, we will show that $r\leq 5$. For $r\geq 6$, consider a parity-check matrix of an $[r+7,r+3,4]$ optimal code, each locality-row of this matrix must have exactly $6$ zeros. Thus, we can puncture this matrix to $12$ of its columns and obtain the parity-check matrix of an optimal $[12,8,4]$ code. To ensure the locality and $n'=6$, the first two rows of this submatrix are isomorphic to those of $H$ in (\ref{length_12_6_only}). Therefore, we only have to show that starting from the parity-check matrix $H$ of an optimal $[12,8,4]$ code, we can't add any other columns into $H$ to obtain a parity-check matrix of an $[r+7,r+3,4]$ code with $r\geq 6$. Otherwise, we obtain a submatrix of form $(\ref{length_12_6_example})$ in Lemma \ref{ban_struture}. Therefore, $r=5$.

When $s\ge 2$, $n'=6$. Due to the arbitrariness of the deleted $s$ locality-rows during the construction of $H'$, we have $r=5$, $k=5s+3$, $n=6s+6$ and $H'$ is the parity-check matrix of a quaternary $[6,3,4]$ MDS code. The following $H$ gives the corresponding optimal code
\begin{flalign}\label{6s+6_5s+3}
H=\left(\begin{array}{c}I_{s+1}\otimes(\begin{array}{cccccc}1&1&1&1&1&1\end{array})\\ \hline \underbrace{(1~1~\ldots~1)}_{s+1}\otimes \Big(\begin{array}{ccccll}0&1&0&\omega^2&\omega^2&1\\0&0&1&\omega^2&1&\omega^2\end{array}\Big) \end{array}\right).
\end{flalign}

\begin{itemize}
  \item For the case $l=\lc k/r\rc+1$
\end{itemize}

Then $n-k=l+1$, which implies that $H_2$ contains one row. According to the construction of $H'$, $H'$ contains two locality-rows covering all the coordinates remained after deletion. Thus $H'$ can be the parity-check matrix of a quaternary $[6,3,4]$ or $[5,2,4]$ or $[4,1,4]$ MDS code. This leads to $4\leq n'\leq 6$.

If $r\mid k$, we set $k=sr$ for some $s>0$. By Lemma \ref{r_mid_k}, we have $(r+1)\mid n$. Since $n=k+l+1=s(r+1)+2$, if $r\ge 2$, we have $(r+1)\nmid n$, which contradicts the above conclusion. When $r=1$, we have $k=s$ and $n=2s+2$. Then $H'$ is a parity-check matrix of a quaternary $[4,1,4]$ code. The following $H$ gives the corresponding optimal code
\begin{flalign}\label{2s+2_s}
H=\left(\begin{array}{c}I_{s+1}~\otimes(\begin{array}{cc}1&1\end{array})\\ \hline \underbrace{(1~\ldots~1)}_{s+1}\otimes(\begin{array}{cc}0&1\end{array}) \end{array}\right).
\end{flalign}

If $r\nmid k$, then $r\ge 2$. Let $k=sr+t$, where $1\leq t\leq r-1$. Similarly, by inequality (\ref{inequality_under_d4}), we have $1\leq t\leq 3$.

If $s=1$ and $t=1$, we have $k=r+1$, $n=r+5$, $r\ge 2$ and $n\leq m_2(3,4)=17$. With the same analysis as that for the case $l=\lc k/r\rc$, we have $n\ne 17$. When $n=16$, we have
\begin{flalign}\label{r+5_r+1}
H=\left(\begin{array}{cccccccccccccccc}
1&1&1&0&0&0&0&\omega^2&\omega&1&\omega&\omega^2&\omega^2&1&\omega^2&1\\
0&0&1&1&1&0&0&\omega&\omega^2&\omega&\omega&0&\omega&\omega&\omega^2&\omega^2\\
0&0&0&0&1&1&1&\omega&\omega&\omega^2&\omega^2&\omega^2&\omega&1&1&\omega\\
\hline
\omega^2&\omega&1&0&1&\omega&\omega^2&1&1&1&1&1&\omega&1&\omega^2&\omega
\end{array}\right).
\end{flalign}
Furthermore, we can puncture the last $0$ to $9$ columns from $(\ref{r+5_r+1})$ to obtain an optimal LRC with parameters $[r+5,r+1,4]$, $2\leq r\leq 11$.

If $s=1$ and $t=2$, we have $k=r+2$, $n=r+6$ and $n\leq m_2(3,4)=17$. Therefore, $3\leq r\leq 11$. When $n=17$, we have
\begin{flalign}\label{r+6_r+2}
H=\left(\begin{array}{ccccccccccccccccc}
1&1&1&1&0&0&0&0&0&\omega&1&\omega^2&\omega^2&\omega^2&\omega&\omega&\omega^2\\
0&0&0&1&1&1&1&0&0&\omega&1&1&\omega&1&\omega^2&1&\omega\\
0&0&1&0&0&0&1&1&1&\omega^2&\omega&\omega&\omega&\omega^2&\omega&1&1\\
\hline
1&\omega&0&1&1&\omega&0&0&\omega^2&1&\omega&\omega&\omega&\omega^2&0&0&1
\end{array}\right).
\end{flalign}
Furthermore, we can puncture the last $0$ to $9$ columns from $(\ref{r+6_r+2})$ to obtain an optimal LRC with parameters $[r+6,r+2,4]$, $3\leq r\leq 11$.

%并且, $H$~的三个局部校验行在~$0$~的部分两两相交至少为~$1$, 即它们的支撑集的补集两两相交至少为~$1$.
%否则, $H$~中存在形如~(\ref{length_12_6_example})~的子矩阵, 进而由引理~\ref{ban_struture}~知~$d\leq 3$, 矛盾.
%进一步地, 由~$[6,3,4]$~MDS~的唯一性, 删去某一局部校验行及其覆盖的列后的矩阵形如(在行变换的意义下同构)
%\begin{flalign}\label{634MDS_pc}
%\widetilde{A}=\left(\begin{array}{ccccll}1&1&1&1&0&0\\1&1&0&0&1&1\\1&\omega&1&\omega&1&\omega \end{array}\right).
%\end{flalign}

If $s=1$ and $t=3$, we have $k=r+3$ and $n=r+7$. Thus, $n'\geq n-(r+1)=6$. Since $n'\leq r+1$, we have $r\ge 5$, $n\ge 12$. Therefore, each locality-row of $H$ has exactly 6 zeros to meet the requirement of parameters. According to the weight distribution of the dual code of $[6,3,4]$ MDS code, complements of the supports of the three locality-rows pairwise intersect in at least two coordinates.

On the other hand, if there exist two locality-rows such that complements of their supports intersect in more than $2$ coordinates, then after deleting arbitrary one locality-row and the columns it covers, the resulting $H'$ has a codeword of weight less than 3, which contradicts that the minimum distance of the dual code of a $[6,3,4]$ MDS code is 4. Therefore, complements of the supports of the three locality-rows pairwise intersect in exact two coordinates.

By regarding each column of $H$ as a point of $PG(3,4)$, we can obtain a point set $\mathcal{S}$. Since $d=4$, $\mathcal{S}$ is a cap in $PG(3,4)$. Consider planes $\pi_1=\{(0,a_1,a_2,a_3):a_1,a_2,a_3\in \mathbb{F}_4\}$, $\pi_2=\{(b_1,0,b_2,b_3):b_1,b_2,b_3\in \mathbb{F}_4\}$, they intersect in line $L=\{(0,0,c_1,c_2):c_1,c_2\in \mathbb{F}_4\}$. Denote $\mathfrak{C}_1=\pi_1\cap\mathcal{S}$, $\mathfrak{C}_2=\pi_2\cap\mathcal{S}$. Since each locality-row of $H$ has exact 6 zeros, $|\mathfrak{C}_1|=|\mathfrak{C}_2|=6=4+2$, i.e., they are hyperovals in $PG(2,4)$. Moreover, since the complements of the supports of the three locality-rows pairwise intersect in exact two coordinates, w.l.o.g, we can suppose that $\mathfrak{C}_1\cap \mathfrak{C}_1=\{A,B\}\subset L$. By Theorem \ref{hyperoval}, $\mathfrak{C}_1\setminus B$ and $\mathfrak{C}_2\setminus B$ are conics in $\pi_1$ and $\pi_2$ respectively. Moreover, both $\pi_1$ and $\pi_2$ have $B$ as their nucleus and they intersect in line $L$ at point $A$ simultaneously. Therefore, according to Lemma~\ref{finite_geometry_1}, $\mathcal{S}$ can be expanded from the incomplete cap $\mathfrak{C}_1\cup\mathfrak{C}_2\cup\mathcal{O}$, and $n\leq3q+2=14$,~$r\leq 7$.

When $r=7$, $n=14$ and $k=10$, the following $H$ gives the corresponding optimal code
\begin{flalign}\label{r+7_r+3}
H=\left(\begin{array}{cccccccccccccc}
1&1&1&1&1&1&0&0&0&0&0&0&1&\omega\\
0&0&0&0&1&1&1&1&1&1&0&0&1&\omega\\
1&1&0&0&0&0&1&1&0&0&1&1&1&\omega\\
\hline
1&\omega&1&\omega&1&\omega&1&\omega&1&\omega&1&\omega&\omega&\omega
\end{array}\right).
\end{flalign}
We can puncture the last 1 or 2 columns from (\ref{r+7_r+3}) to obtain LRCs with parameters $[13,9,4]$ and $[12,8,4]$.
%Besides, from the connection we build above, columns $(1,1,1,\omega)^T$, $(\omega,\omega,\omega,\omega)^T$ can be regarded as points in the plane $\pi_3:=\{(a,a,b,c):~a\in\mathbb{F}_4^*,~b,c\in\mathbb{F}_4\}$ through line $L$.

If $s\ge 2$, recall that $k=sr+t$, where $1\leq t\leq 3$ and $r>t$. According to the value of $t$, we divide our discussion into the following $3$ cases.

If $t=3$, by inequality (\ref{inequality_under_d4}), we have $n'=6$ and $H'$ can only be the parity-check matrix of a quaternary $[6,3,4]$ MDS code. According to the minimum distance of the dual code of the quaternary $[6,3,4]$ MDS code, we know that the supports of the rows of $H'$ pairwise intersect on at least two coordinates. Meanwhile, we have $\gamma=n-n'=s(r+1)$, which means the support sets of the $s$ deleted locality-rows are pairwise disjoint with size $r+1$, a contradiction. Therefore, optimal LRCs with such type of parameter don't exist.

If $t=2$, by inequality (\ref{inequality_under_d4}), we have $n'=5,6$ and $H'$ can be the parity-check matrix of a quaternary $[5,2,4]$ or $[6,3,4]$ MDS code. With a similar analysis as above, one can also show that optimal LRCs with such type of parameter don't exist.

If $t=1$, by inequality (\ref{inequality_under_d4}), we have $4\leq n'\leq 6$ and $H'$ can be the parity-check matrix of a quaternary $[4,1,4]$, $[5,2,4]$ or $[6,3,4]$ MDS code. Since $\gamma=n-n'=s(r+1)+4-n'$, when $n'=4$, the supports of the $s$ deleted locality-rows are pairwise disjoint with size $r+1$.
%Due to the arbitrariness of these $s$ locality-rows, all the locality-rows of $H$ must have weight exactly $r+1$.
Similarly, when $n'=5$ or $6$, one can show that there exist at most 1 or 2 coordinates covered by different members of the $s$ deleted locality-rows. Moreover, by the expression of $\gamma$, every two locality-rows of $H'$ intersect on at most two coordinates.
 %The total sum of the intersections of all pairs of locality-rows is at most $(r+1)(n-k-1)-n=2r-2$.

For the case when $s=2$, $k=2r+1$, $n=2(r+1)+4$. According to value of $\gamma$, there exist two coordinates covered by the two deleted locality-rows. Therefore, the length of $H$ is at least $4(r+1)-2\times 6$, then we have $4(r+1)-2\times 6\leq 2r+6$. Thus, $2\leq r\leq 7$. The following $H_r$s give the corresponding optimal codes when $2\leq r\leq 7$.

\begin{flalign}\label{2r+6_2r+1_1}
H_2=\left(\begin{array}{cccccccccc}
1&1&1&0&0&0&0&0&0&0\\
0&0&1&1&1&0&0&0&0&0\\
0&0&0&0&0&1&1&1&0&0\\
0&0&0&0&0&0&0&1&1&1\\
\hline
1&\omega&1&1&\omega&1&\omega&1&1&\omega
\end{array}
\right),~H_3=\left(\begin{array}{cccccccccccc}
1&1&1&1&0&0&0&0&0&0&0&0\\
0&0&1&1&1&1&0&0&0&0&0&0\\
0&0&0&0&0&0&1&1&1&1&0&0\\
0&0&0&0&0&0&0&0&1&1&1&1\\
\hline
1&\omega&1&\omega&1&\omega&1&\omega&1&\omega&1&\omega
\end{array}
\right),
\end{flalign}
\begin{flalign}\label{2r+6_2r+1_2}
H_4=\left(\begin{array}{cccccccccccccc}
1&1&1&1&1&0&0&0&0&0&0&0&0&0\\
0&0&0&1&1&1&1&1&0&0&0&0&0&0\\
1&0&0&0&0&0&0&0&1&1&1&1&0&0\\
0&0&0&0&0&1&1&0&0&0&1&1&1&1\\
\hline
1&1&\omega&1&\omega&1&\omega&1&1&\omega&1&\omega&1&\omega
\end{array}
\right),~
H_5=\left(\begin{array}{cccccccccccccccc}
1&1&1&1&1&1&0&0&0&0&0&0&0&0&0&0\\
0&0&0&0&1&1&1&1&1&1&0&0&0&0&0&0\\
1&0&0&0&0&0&0&0&1&0&1&1&1&1&0&0\\
0&1&0&0&0&0&0&0&0&1&0&0&1&1&1&1\\
\hline
1&1&1&\omega&1&\omega&1&\omega&1&1&1&\omega&1&\omega&1&\omega
\end{array}
\right),
\end{flalign}
\begin{flalign}\label{2r+6_2r+1_3}
H_6=\left(\begin{array}{cccccccccccccccccc}
1&1&1&1&1&1&1&0&0&0&0&0&0&0&0&0&0&0\\
0&0&0&0&0&1&1&1&1&1&1&1&0&0&0&0&0&0\\
1&0&1&0&0&0&0&0&0&0&1&0&1&1&1&1&0&0\\
0&1&0&0&0&0&0&0&0&1&0&1&0&0&1&1&1&1\\
\hline
1&1&\omega&1&\omega&1&\omega&1&\omega&1&1&\omega&1&\omega&1&\omega&1&\omega
\end{array}
\right),
\end{flalign}
\begin{flalign}\label{2r+6_2r+1_4}
H_7=\left(\begin{array}{cccccccccccccccccccc}
1&1&1&1&1&1&1&1&0&0&0&0&0&0&0&0&0&0&0&0\\
0&0&0&0&0&0&1&1&1&1&1&1&1&1&0&0&0&0&0&0\\
1&0&1&0&0&0&0&0&0&0&1&0&1&0&1&1&1&1&0&0\\
0&1&0&1&0&0&0&0&0&0&0&1&0&1&0&0&1&1&1&1\\
\hline
1&1&\omega&\omega&1&\omega&1&\omega&1&\omega&1&1&\omega&\omega&1&\omega&1&\omega&1&\omega
\end{array}
\right).
\end{flalign}

%For the case $s=2$ and $r\ge 8$. Since $\gamma=2(r+1)+4-n'$, we know that the deleted $2$ locality-rows share at most two common coordinates. Noted that there exists a locality-row of weight $r+1\geq 9$, hence, there exist at least $3$ coordinates covered only by this locality-low. Therefore, w.l.o.g., we can assume that the $3$ columns corresponding to $3$ coordinates all have the form $(1,0,0,0,x_i)$ ($1\leq i\leq 3$), where $x_i\in \mathbb{F}_4$. These $3$ columns are linearly dependent, this contradicts the fact that $d=4$. Thus, optimal LRCs with such parameters don't exist.

For the case $s\ge 3$. If there exists an $H'$ with $n'=6$, since every two rows of $H'$ share at least $2$ common coordinates, by $\gamma$, $n'$ can only be $6$. Thus, every locality-row has weight $r+1$ and every two of them share exactly two common coordinates. It is easy to show that we can obtain at most $2+(r-1)(s+2)$ columns, therefore, we have $2+(r-1)(s+2)\geq n$ and $r\ge 5$.
%and these common coordinates contributed ${s\choose 2}\cdot2\ge 6$ intersections. Since the total intersection is at most $(n-k-1)\cdot(r+1)-n=2r-2$, thus we have $6\leq 2r-2$, i.e. $r\geq 4$.
When $r\ge 5$, assume that the first two locality-rows remained after the deleting process during the construction of $H'$, then these two locality-rows have two common coordinates. Moreover, there exists another locality-row sharing at least one common coordinate with one of the first two locality-rows. By the arbitrariness of the rows deleted to construct $H'$, we can take these $3$ locality-rows during the deleting process, then $\gamma\leq s(r+1)-3$ which contradicts the fact that $\gamma\ge s(r+1)-2$. Therefore, $n'\neq 6$. With a similar discussion, we have $n'\neq5$ either. Therefore, we have $n'=4$ and the supports of any two locality-rows are pairwise disjoint. Thus, $r=1$ and this contradicts the assumption that $r\nmid k$.
%since $r\geq 4$, by a similar discussion as that for the case $s=2$ and $r\ge 8$, we have $3$ columns linearly dependent. This contradicts the fact that $d=4$.
In conclusion, optimal LRCs with such parameters don't exist.

\begin{itemize}
  \item For the case $l=\lc k/r\rc+2$
\end{itemize}

%Then $n-k=l$, which implies that all rows of $H$ are locality-rows. Therefore, the union of the supports of any $l-1$ rows in $H$ can not cover $[n]$.

If $r\mid k$, by Lemma \ref{r_mid_k}, the supports of all the locality-rows in $H$ must be pairwise disjoint. When $n'=5,6$, the support of $H'$ must intersect in some coordinates according to their weight distribution. Thus $n'=4$ and the supports of the locality-rows are pairwise disjoint with weight exactly $r+1$. Then one can easily find two columns covered by a same locality-row linearly dependent, which contradicts the fact that $d=4$. Therefore, $r\nmid k$.

Assume $k=sr+t$, where $1\leq t\leq 3$, with the same proof as the case when $l=\lc k/r\rc+1$ in Subsection \ref{d=3}, we can obtain $s=1$ or $s\ge 2$ with $t=1$ by discussing the relationship between the intersection of coordinates of the $s$ deleted locality-rows and the intersection of any two rows of $H'$.

\subsection{$5\leq d\leq 8$ and $H'$ contains more than three rows}
%
%First, we need the following lemma to help us determine the upper bound on the minimum distance for optimal LRCs.
%\begin{lem}\cite{HSX2019}\label{singleton_defect}
%Let $\mathcal{C}$ be an optimal $q$-ary $(n,k,r)$-LRC with minimum distance $d$ and dimension $k>r\ge 1$, then
%\begin{align*}
%d\leq\left\{\begin{array}{ll}q,&\text{if }r\nmid (k-1)\text{;}\\2q,&\text{if }r\mid (k-1).\end{array}\right.
%\end{align*}
%\end{lem}
In this subsection, $H'$ is a full rank parity-check matrix of a quaternary $[n',1,n']$ $(n'\ge 5)$ code. Since $H'$ is an $(n'-1)\times n'$ matrix and $\mathcal{C'}$ has $d'=d$, we have $n'=d'=d$ and $m'=n'-1=d-1=n-k-\lc k/r \rc+1$. Besides, the number of the deleted columns is $\gamma=n-n'=k+\lc k/r\rc-2$. Since each of the locality-rows has weight at most $r+1$, then
\begin{flalign}\label{page_14}
\gamma=&k+\lc k/r\rc-2\le (\lc k/r\rc-1)(r+1),\nonumber\\
\text{i.e.                } ~~~~~~~~&k-r\cdot\lc k/r\rc+r\le 1.
\end{flalign}

If $r\mid k$, (\ref{page_14}) implies that $r=1$. By Lemma \ref{r_mid_k}, we have $(r+1)\mid n$. Then $n=2l$ and $d=n-k-\lc k/r\rc +2=2(l-k+1)$. Since $r\mid (k-1)$, we have $5\leq d=2(l-k+1)\leq2q=8$ by Lemma \ref{singleton_defect}. Hence $2\leq l-k\leq 3$. Moreover, since the supports of any two locality-rows of the parity-check matrix are pairwise disjoint and have size $r+1$, $H$ has the form as $(\ref{H})$.

\begin{itemize}
  \item For the case $l=k+2$
\end{itemize}

We have $n=2k+4$ and $d=6$. Applying Lemma \ref{keylemma}.1 with $u=n-k-l=2$, we have $l=k+2\leq\frac{4^2-1}{4-1}=5$. Hence, $k=2$ or $k=3$. The following two parity-check matrices give optimal constructions for $k=2$ and $k=3$ respectively,
\begin{flalign}\label{2k+4}
H=\left(\begin{array}{cccccccc}1&1&0&0&0&0&0&0\\ 0&0&1&1&0&0&0&0\\0&0&0&0&1&1&0&0\\0&0&0&0&0&0&1&1\\ \hline 0&1&0&1&0&1&0&0 \\ 0&\omega&0&1&0&0&0&1 \end{array}\right),~~~~
H=\left(\begin{array}{cccccccccc}1&1&0&0&0&0&0&0&0&0\\ 0&0&1&1&0&0&0&0&0&0\\0&0&0&0&1&1&0&0&0&0\\0&0&0&0&0&0&1&1&0&0\\0&0&0&0&0&0&0&0&1&1\\ \hline 0&1&0&1&0&\omega&0&\omega&0&0 \\0&\omega&0&1&0&1&0&0&0&1 \end{array}\right).
\end{flalign}

\begin{itemize}
  \item For the case $l=k+3$
\end{itemize}

We have $n=2k+6$ and $d=8$. Applying Lemma \ref{keylemma}.3 with $u=3$, we have $l=k+3\leq m_2(2,4)=q+2=6$. Hence, $k=2$ or $k=3$. The following two parity-check matrices give optimal constructions for $k=2$ and $k=3$ respectively,

\begin{flalign}\label{2k+6}
H=\left(\begin{array}{cccccccccc}1&1&0&0&0&0&0&0&0&0\\ 0&0&1&1&0&0&0&0&0&0\\0&0&0&0&1&1&0&0&0&0\\0&0&0&0&0&0&1&1&0&0\\0&0&0&0&0&0&0&0&1&1\\ \hline 0&1&0&1&0&\omega&0&\omega&0&0 \\0&\omega&0&1&0&1&0&0&0&1\\0&0&0&1&0&0&0&\omega&0&1 \end{array}\right),
~~~~
H=\left(\begin{array}{cccccccccccc}1&1&0&0&0&0&0&0&0&0&0&0\\ 0&0&1&1&0&0&0&0&0&0&0&0\\0&0&0&0&1&1&0&0&0&0&0&0\\0&0&0&0&0&0&1&1&0&0&0&0\\0&0&0&0&0&0&0&0&1&1&0&0\\0&0&0&0&0&0&0&0&0&0&1&1\\ \hline 0&1&0&1&0&\omega&0&\omega&0&0&0&0 \\0&\omega&0&1&0&1&0&0&0&1&0&0\\0&0&0&1&0&0&0&\omega&0&1&0&1 \end{array}\right).
\end{flalign}

If $r\nmid k$, then $r\ge 2$. Set $k=sr+t$, where $1\leq t\leq r-1$. By (\ref{page_14}), we have $t=1$. This implies $k=sr+1$ and $\gamma=s(r+1)$. Since we need to delete $\lc k/r\rc-1=s\ge1$ locality-rows, each of which has weight at most $r+1$. Thus, these deleted $s$ rows must have disjoint supports with weight $r+1$. Next, we divide our discussion into two cases: $s=1$ and $s\ge 2$.

\textbf{Case} $\mathbf{s=1}$: We have $k=r+1$ and $d=n-k=n'> 4$. By Hao et al.'s arguments in Section III.B of \cite{HXC2017}, $\mathcal{C}$ and $\mathcal{C}^\bot$ are both near MDS codes. Now, we use the following lemma to determine the range of parameter $d$.

\begin{lem}\cite[Proposition 6.2]{DSL1995}\label{near_mds_1}
If $\mathcal{C}$ is an $[n,k,n-k]$ near MDS code over $\mathbb{F}_q$, $q>3$, with $k\ge3$, then $n\leq 2q+k-2$.
\end{lem}

Therefore, we have $d=n-k\le 6$, the minimum distance of $H$ is $d^\bot=r+1=k\leq 6$. Thus all the possible parameters of quaternary near MDS codes are $[n,k,n-k]$, for $3\le k\le 6$ and $5\le n-k\le 6$.

There exist optimal LRCs of locality $r=k-1$ with all the above parameters. The following $H$ gives a parity-check matrix of an optimal quaternary $[12,6,6]$ LRC with $r=5$:
\begin{flalign}\label{near_mds}
H=\left(\begin{array}{llllllllllll}
1~&1~&1~&1~&1~&1~&0~&0~&0&0&0&0\\
0&0&0&0&0&0&1&1&1&1&1&1\\
\hline
0&1&0&0&\omega&1&0&0&1&0&1&\omega^2\\
0&0&\omega^2&0&1&0&0&0&1&\omega&1&\omega\\
0&0&0&\omega^2&\omega&\omega^2&0&0&\omega^2&\omega^2&0&1\\
0&0&0&0&\omega&\omega&0&\omega^2&\omega^2&1&\omega&\omega
\end{array}\right).
\end{flalign}
All the optimal LRCs with other parameters can be obtained by puncturing or shortening this $[12,6,6]$ code. Moreover, their localities satisfy $r=k-1$ as well.

\textbf{Case} $\mathbf{s\ge 2}$: Since the deleted $s$ locality-rows can be chosen arbitrarily and they have disjoint supports of weight $r+1$, we know that all locality-rows of $H$ are pairwise disjoint, then $H$ has the form $(\ref{H})$. Hence, $(r+1)\mid n$ and we have $n=l(r+1)$. Then, $d=(r+1)(l-s)$. Since $r\mid (k-1)$, by Lemma \ref{singleton_defect}, we have $d=(r+1)(l-s)\leq 2q=8$. Combining with $r\ge2$, we also have $l-s\le 2$.

\begin{itemize}
  \item For the case $l=s+1$
\end{itemize}

We have $5\leq d=r+1\leq 8$, i.e. $4\leq r\leq 7$. By Lemma \ref{abount_H''}, after deleting any $\lc k/r \rc- 2$ locality-rows and all the columns they covered, we obtain $[2(r+1),r+1,r+1]$ almost MDS code. Applying Lemma $\ref{amds}$ with $q=4$, we have $r+1<7$. Therefore, $4\leq r\leq 5$.

When $r=4$, we have $k=4s+1$, $n=5s+5$ and $d=5$. Applying Lemma \ref{keylemma}.1 with $u=n-k-(s+1)=3$, we have $(s+1){5\choose2}\leq\frac{4^3-1}{4-1}$. Thus, $s\leq1$, a contradiction.

When $r=5$, we have $k=5s+1$, $n=6s+6$ and $d=6$. Applying Lemma \ref{keylemma}.2 with $u=n-k-(s+1)=4$, we have $(s+1){6\choose2}+2\cdot{5\choose 3}+{5\choose 2}\leq\frac{4^4-1}{4-1}$. Thus, $s\leq2$. If $s=2$, we have $n=18$, $k=11$ and $d=6$. Next, we will show that $[18,11,6]$ codes with locality $r=5$ do not exist.

Suppose that there exists such a code with parity-check matrix $H_0$, by deleting any one locality-row and its corresponding columns from $H_0$, the remaining submatrix is a parity-check matrix of a $[12,6,6]$ code according to Lemma \ref{abount_H''}. W.l.o.g., we suppose that $H_0$ has the following form:
\begin{flalign}\label{[12,6,6],WLOG}
H_0=\left(\begin{array}{cccccccccccccccccc}1&1&1&1&1&1&0&0&0&0&0&0&0&0&0&0&0&0\\0&0&0&0&0&0&1&1&1&1&1&1&0&0&0&0&0&0\\0&0&0&0&0&0&0&0&0&0&0&0&1&1&1&1&1&1\\ \hline
0&1&0&0&0&\multirow{4}{*}{$a$} &0&\multirow{4}{*}{$b_1$} &\multirow{4}{*}{$b_2$} &\multirow{4}{*}{$b_3$} & \multirow{4}{*}{$b_4$}&\multirow{4}{*}{$b_5$} &0&\multirow{4}{*}{$c_1$} &\multirow{4}{*}{$c_2$} &\multirow{4}{*}{$c_3$} & \multirow{4}{*}{$c_4$}& \multirow{4}{*}{$c_5$}\\ 0&0&1&0&0& &0& & & & & &0& & & & & \\ 0&0&0&1&0& &0& & & & & &0& & & & & \\ 0&0&0&0&1& &0& & & & & &0& & & & & \end{array}\right),
\end{flalign}
where $a,b_1,\ldots,b_5,c_1,\ldots,c_5\in \mathbb{F}_4^4$. We have the following claim.

\begin{claim}\label{[12,6,6] unique}
Assume that $S_1$, $S_2$ are both parity-check matrices of a quaternary $[12,6,6]$ code, which have the following forms, respectively
\begin{flalign}\label{[12,6,6]}
S_1=\left(\begin{array}{cccccccccccc}1&1&1&1&1&1&0&0&0&0&0&0\\0&0&0&0&0&0&1&1&1&1&1&1\\
0&1&0&0&0&\multirow{4}{*}{$\mathbf{a}$} &0&\multirow{4}{*}{$\mathbf{b_1}$} &\multirow{4}{*}{$\mathbf{b_2}$} &\multirow{4}{*}{$\mathbf{b_3}$} & \multirow{4}{*}{$\mathbf{b_4}$}&\multirow{4}{*}{$\mathbf{b_5}$}\\ 0&0&1&0&0& &0& & & & &  \\ 0&0&0&1&0& &0& & & & &  \\ 0&0&0&0&1& &0& & & & &  \end{array}\right),
~S_2=\left(\begin{array}{cccccccccccc}1&1&1&1&1&1&0&0&0&0&0&0\\0&0&0&0&0&0&1&1&1&1&1&1\\
0&1&0&0&0&\multirow{4}{*}{$\mathbf{a}$} &0&\multirow{4}{*}{$\mathbf{c_1}$} &\multirow{4}{*}{$\mathbf{c_2}$} &\multirow{4}{*}{$\mathbf{c_3}$} & \multirow{4}{*}{$\mathbf{c_4}$}&\multirow{4}{*}{$\mathbf{c_5}$}\\ 0&0&1&0&0& &0& & & & &  \\ 0&0&0&1&0& &0& & & & &  \\ 0&0&0&0&1& &0& & & & &  \end{array}\right),
%S=\left(\begin{array}{cccccccccccc}1&1&1&1&1&1&0&0&0&0&0&0\\0&0&0&0&0&0&1&1&1&1&1&1\\
%0&1&0&0&0&\multirow{4}{*}{$a$} &0&\multirow{4}{*}{$b_1$} &\multirow{4}{*}{$b_2$} &\multirow{4}{*}{$b_3$} & \multirow{4}{*}{$b_4$}&\multirow{4}{*}{$b_5$}\\ 0&0&1&0&0& &0& & & & &  \\ 0&0&0&1&0& &0& & & & &  \\ 0&0&0&0&1& &0& & & & &  \end{array}\right),
\end{flalign}
where $\mathbf{a},\mathbf{b_1},\ldots,\mathbf{b_5},\mathbf{c_1},\ldots,\mathbf{c_5}\in \mathbb{F}_4^4$. Then $b_i=\lambda c_{j_i}$, where $1\leq i\leq 5$, $\lambda\in\mathbb{F}_4^{\ast}$ and $\{j_1,\ldots,j_5\}=\{1,\ldots,5\}$.

%where $a$ is a fixed vector in $\mathbb{F}_4^4$ and $b_1,\ldots,b_5\in\mathbb{F}_4^4$ can be arbitrarily chosen. If there exists another set of vectors $\{c_1,\ldots,c_5\}$ satisfying $(\ref{[12,6,6]})$ with the same $a$, then $b_i=\lambda c_{j_i}$, where $1\leq i\leq 5$, $\lambda\in\mathbb{F}_4^{\ast}$, $\{j_1,\ldots,j_5\}=\{1,\ldots,5\}$.
\end{claim}
%and the dual code of quaternary $[12,6,6]$ code is also a quaternary $[12,6,6]$ code
\begin{proof}[Proof of Claim \ref{[12,6,6] unique}]
In \cite{DSL1994}, the authors showed that the quaternary $[12,6,6]$ code is unique. Thus, the parity-check matrix of each quaternary $[12,6,6]$ can be obtained from $S$ by the following operations: row linear transformations, column permutations and multiplications of columns by an element of $\mathbb{F}_4^*$. Therefore, w.l.o.g., we can assume that each parity-check matrix $S'$ of the quaternary $[12,6,6]$ code has the form $S'=LSR$, for invertible matrices $L=(\ell_{i,j})_{1\leq i,j\leq 6}$ and $R=diag\{r_1,\ldots,r_{12}\}$, where $\ell_{i,j}\in\mathbb{F}_4$ and $r_t\in\mathbb{F}_4^\ast$ ($1\leq t\leq 12$).

If there exists another parity-check matrix $S'=(s_{i,j})_{1\leq i\leq 6,1\leq j\leq 12}$ of the form (\ref{[12,6,6]}) with $\{b_1,\ldots,b_5\}$ replaced by $\{c_1,\ldots,c_5\}$ and the same $a$, we can obtain the following equalities:
\begin{flalign*}
r_1\ell_{i,1}=s_{i,1},~1\leq i\leq 6;\\
r_2(\ell_{i,1}+\ell_{i,3})=s_{i,2},~1\leq i\leq 6;\\
r_3(\ell_{i,1}+\ell_{i,4})=s_{i,3},~1\leq i\leq 6;\\
r_4(\ell_{i,1}+\ell_{i,5})=s_{i,4},~1\leq i\leq 6;\\
r_5(\ell_{i,1}+\ell_{i,6})=s_{i,5},~1\leq i\leq 6;\\
r_7\ell_{i,2}=s_{i,7},~1\leq i\leq 6.
\end{flalign*}
According to the values of $s_{i,j}$, $L$ has the following form
\begin{flalign*}
L=\left(\begin{array}{cccccc}
r_1^{-1}&0&r_1^{-1}+r_2^{-1}&r_1^{-1}+r_3^{-1}&r_1^{-1}+r_4^{-1}&r_1^{-1}+r_5^{-1}\\
0&r_7^{-1}&0&0&0&0\\
0&0&r_2^{-1}&0&0&0\\
0&0&0&r_3^{-1}&0&0\\
0&0&0&0&r_4^{-1}&0\\
0&0&0&0&0&r_5^{-1}
\end{array}
\right).
\end{flalign*}
Denote $b_{i,j}$, $1\leq i\leq4$, $1\leq j \leq 5$ as the $i_{th}$ coordinate of vector $b_j$, then we can obtain:
\begin{flalign*}
(r_1^{-1}+r_2^{-1})b_{1,j}+(r_1^{-1}+r_3^{-1})b_{2,j}+(r_1^{-1}+r_4^{-1})b_{3,j}+(r_1^{-1}+r_5^{-1})b_{4,j}=0,\text{ for }1\leq j\leq 5.
\end{flalign*}
Take $(r_1^{-1}+r_t^{-1})$ ($2\leq t\leq 5$) as variables and $b_{i,j}$ as coefficients, the identities above can be considered as a system of linear equations. If these equations have non-zero solutions, then the rank of coefficient matrix is less than 4, which means that there exist $t_1,\ldots,t_4\in[5]$ such that $b_{t_1}$ is a linear combination of $b_{t_2},b_{t_3},b_{t_4}$. Then, consider the $5$ columns consisting of the $7_{th}$ and the $(7+t_i)_{th}$ ($1\leq i\leq 4$) column of $S$, they are linearly dependent, which contradicts the fact that $S$ is the parity-check matrix of a code with minimum distance 6. Therefore, we have $r_1=r_i$ for $2\leq i\leq 5$, and $L=diag(r_1^{-1},r_7^{-1},r_1^{-1},r_1^{-1},r_1^{-1},r_1^{-1})$. One can also show that $r_i=r_7$, for $8\leq i\leq 12$. Therefore, by $\sum_{j=1}^{6}l_{k,j}s_{j,i}=s_{1,i}$, we have $b_i=r_1 r_7^{-1}c_i$ for $1\leq i\leq 5$.
\end{proof}

Therefore, together with Lemma \ref{abount_H''}, $H_0$ in (\ref{[12,6,6],WLOG}) must satisfy additional conditions $b_i=\lambda c_{j_i}$ ($1\leq i\leq 5$), for some $\lambda\in\mathbb{F}_4^{\ast}$ and $\{j_1,\ldots,j_5\}=\{1,2,\ldots,5\}$. This contradicts the fact that $H_0$ is the parity-check matrix of a code with minimum distance 6. Therefore, $[18,11,6]$ codes with locality $r=5$ do not exist.

\begin{itemize}
  \item For the case $l=s+2$
\end{itemize}

We have $d=2(r+1)\leq 8$, i.e. $r\le 3$. By Lemma \ref{abount_H''}, after deleting any $\lc k/r \rc- 2$ locality-rows and all the columns they covered, we obtain $[3(r+1),r+1,2r+2]$ almost MDS code. Applying Lemma $\ref{amds}$ with $q=4$, we have $2r+2<7$. Therefore, $r=2$.

When $r=2$, we have $k=2s+1$, $n=3s+6$ and $d=6$. Applying Lemma \ref{keylemma}.2 with $u=n-k-(s+2)=3$, we have $(s+2)\cdot{3\choose 2}+2\cdot{3\choose 3}\leq \frac{4^3-1}{4-1}=21$. Therefore, $2\leq s\leq 4$. When $s=4$, following the idea of the proof of Lemma \ref{keylemma}, we convert the columns of $H$ into points of the $PG(2,4)$. Since $d=6$, any at most 5 columns of $H$ are linearly independent. Therefore, the points generated by 2 columns covered by different locality rows are different, and the point generated by 2 columns covered by one locality-row and the point generated by 3 columns covered by one locality-row are different. Meanwhile, by $l=6$, there are only $21-6\cdot{3\choose 2}=3$ points left in $PG(2,4)$ which can be generated by $3$ columns covered by one locality-row. By pigeonhole principle, there exist two points generated by two different sets of 3 columns, then there exist $5$ columns linearly dependent, which leads a contradiction. The following parity-check matrix gives optimal constructions for $s = 2$ and $s = 3$ respectively.
\begin{flalign}\label{3s+6_1}
H=\left(\begin{array}{llllllllllll}
1~~&1&1~~&0~~&0&0~~&0~~&0&0~~&0~~&0~~&0~\\
0&0&0&1&1&1&0&0&0&0&0&0\\
0&0&0&0&0&0&1&1&1&0&0&0\\
0&0&0&0&0&0&0&0&0&1&1&1\\
\hline
0&\omega^2&1&0&\omega^2&1&0&\omega^2&1&0&0&0\\
0&0&\omega&0&\omega&\omega&0&\omega&0&0&1&0\\
0&1&1&0&1&0&0&\omega^2&0&0&0&1
 \end{array}\right),
\end{flalign}
\begin{flalign}\label{3s+6_2}
H=\left(\begin{array}{lllllllllllllll}
1~~&1&1~~&0~~&0&0~~&0~~&0&0~~&0~~&0~~&0~~&0~~&0~~&0\\
0&0&0&1&1&1&0&0&0&0&0&0&0&0&0\\
0&0&0&0&0&0&1&1&1&0&0&0&0&0&0\\
0&0&0&0&0&0&0&0&0&1&1&1&0&0&0\\
0&0&0&0&0&0&0&0&0&0&0&0&1&1&1\\
\hline
0&\omega^2&1&0&\omega^2&1&0&\omega^2&1&0&0&0&0&1&\omega\\
0&0&\omega&0&\omega&\omega&0&\omega&0&0&1&0&0&0&\omega^2\\
0&1&1&0&1&0&0&\omega^2&0&0&0&1&0&1&\omega^2
\end{array}\right).
\end{flalign}

\section{Concluding Remarks}\label{section_5}
In this paper, we determine all possible parameters of the optimal quaternary LRCs and give corresponding explicit constructions. During the discussion, we use tools from finite geometry to give some necessary conditions for the existence of optimal quaternary LRCs. We believe that it is helpful to establish a connection between the structures of parity-check matrices for LRCs and combinatorial structures, which will be reported in our future work.\\

\textbf{Comment: }The Chinese version of this paper will appear in SCIENTIA SINICA Mathematica (DOI: 10.1360/SSM-2022-0041).

\bibliographystyle{abbrv}
\bibliography{F4_LRC_2022}
%\begin{appendices}
%\section{aa}
%
%\section{bb}
%cc
%\end{appendices}

\end{document}